\documentclass{article}    
    
    \usepackage{amsmath,amsthm,amssymb,amsfonts,bm,mathtools}
    \usepackage{graphicx}
    \usepackage{enumitem}
    \usepackage{hyperref}
    \usepackage{bbm}
    \usepackage{algorithm2e}
    \usepackage{booktabs,makecell}
    \usepackage{cleveref}
    \usepackage{tikz}
    \usetikzlibrary{arrows.meta,calc}
    \usepackage{authblk}

    \hypersetup{
      colorlinks   = true,
      urlcolor     = blue,
      linkcolor    = blue,
      citecolor    = blue
    }
    
    \theoremstyle{plain}
    \newtheorem{theorem}{Theorem}[section]
    \newtheorem{proposition}[theorem]{Proposition}
    
    \newtheorem{lemma}[theorem]{Lemma}

    \theoremstyle{definition}
    \newtheorem{remark}[theorem]{Remark}
    \newtheorem{assumption}[theorem]{Assumption}

    \newcommand{\E}{\mathbb{E}}
    \newcommand{\R}{\mathbb{R}}
    \newcommand{\Var}{\mathrm{Var}}
    \newcommand{\1}{\mathbf{1}}
    \newcommand{\PP}{\mathbb{P}}

\begin{document}

\title{Change Point Detection and Mean-Field Dynamics of Variable Productivity Hawkes Processes}
\date{December 2025} 
\author[1*]{Conor Kresin}
\author[1]{Boris Baeumer}
\author[2]{Sophie Phillips}
\affil[1]{University of Otago, New Zealand}
\affil[2]{University of California, Los Angeles}
\affil[*]{Corresponding author, conor.kresin@otago.ac.nz }

\maketitle

\begin{abstract}
    Many self-exciting systems change because endogenous amplification, as opposed to exogenous forcing, varies. We study a Hawkes process with fixed background rate and kernel, but piecewise time-varying productivity. For exponential kernels we derive closed-form mean-field relaxation after a change and a deterministic surrogate for post-change Fisher information, revealing a boundary layer in which change time information localises and saturates, while post-change level information grows linearly beyond a short transient. These results motivate a Bayesian change point procedure that stabilizes inference on finite windows. We illustrate the method on invasive pneumococcal disease incidence in The Gambia, identifying a decline in productivity aligned with pneumococcal conjugate vaccine rollout.
    \end{abstract}
    

    \section{Introduction}
    
    The Hawkes point process is a canonical model for clustered events across seismology, finance, and contagion \cite{bacryHawkesProcessesFinance2015,hawkes1971spectra,ogata2006space}. It represents the conditional rate as the sum of an exogenous background and an endogenous offspring component triggered by prior events. Variants such as the recursive Hawkes process \cite{schoenberg2019recursive} and the HawkesN model \cite{rizoiu2018sir} have been used to model the spread of infectious diseases including Chlamydia \cite{paik2022nonparametric}, meningococcal disease \cite{meyerSpacetimeConditionalIntensity2012}, and COVID-19 \cite{chiang2020hawkes,garetto2021time,phillipsDetectionSurgesSARSCov22025}. 
    
    We focus on systems whose dynamics change because endogenous amplification changes. In Hawkes models, this amplification is the productivity (equivalently, branching ratio) $\kappa(t)$, the expected number of offspring per event. In epidemiological settings $\kappa(t)$ parallels a time-varying reproduction factor, and its estimation is central to surveillance and decision making following interventions or behavioural shifts \cite{bertozziChallengesModelingForecasting2020b,bootsmaEffectPublicHealth2007,chiang2020hawkes, wallingaDifferentEpidemicCurves2004}. Point process and related models have long been used to estimate reproduction numbers using methods such as spline based Poisson regression \cite{coriNewFrameworkSoftware2013,hong2020estimation}. While the classical Hawkes model imposes $\kappa<1$ for stationarity \cite{hawkes1971spectra}, finite observation windows in epidemics naturally include explosive regimes with $\kappa>1$; such explosive regimes have been explored as an early warning signal for surges \cite{phillipsDetectionSurgesSARSCov22025}. Our formulation can accommodate explosive regimes, but in this paper we restrict change point inference to subcritical segments unless change times are known or externally constrained; see Section~\ref{sec:stability}.
    
    Nonstationary Hawkes variants allow either the background rate or the productivity to vary over time \cite{briz2025self,harte2014etas,kumazawa2014nonstationary,omi2017hawkes,roueff2019time,zhuang2019semiparametric}. Parameter estimation for such processes is closely related to the dual problem of change point detection, where the goal is to localize structural breaks in Hawkes dynamics. Online approaches adapt CUSUM or score statistics and often assume known pre  and post-change parameters \cite{wang2023sequential,zhang2023online}; sequential likelihood ratio tests have also been proposed, though they can be computationally intensive due to repeated expectation maximisation (EM) estimation \cite{li2017detecting}. Offline methods include penalized dynamic programming in discrete time \cite{wangDetectingAbruptChanges2024} and continuous time dynamic programming that mitigates memory effects via time scaling \cite{dion2024multiple}. Bayesian formulations appear both as state space models that infer latent states and time-varying productivity \cite{koyamaEstimatingTimevaryingReproduction2021} and as predictive deviation schemes that flag departures \cite{zhang2024conjugate}. Much of this literature optimizes for localizing when changes occur. By contrast, our interest is to characterize what changes and how a new regime evolves once a change has occurred, both locally and globally.

    We adopt a piecewise specification for $\kappa(t)$ while keeping the background rate and kernel parameters fixed across segments. This directly targets endogenous amplification and reduces confounding between baseline intensity and productivity. The choice is substantive rather than cosmetic: Proposition~\ref{prop:no-reparam} shows that time variation in $\kappa(t)$ cannot, path wise, be absorbed into a deterministic time-varying background.
    
    Variable productivity (or more generally nonstationary Hawkes processes) do not necessarily imply the existence of change points. We write $\kappa(t)=f(t;\eta(t))$ for a parametric form $f$. A variable productivity model allows $\kappa(t)$ to evolve with $\eta$ constant in time. In contrast, a change point model is needed when $\eta(t)$ is not constant in time, i.e. the parameter path is piecewise constant in time (reference Figure~\ref{fig:kappa-cp-illustration}). change points are therefore jumps in parameters, not necessarily discontinuities in $\kappa(t)$. In many cases,  $\kappa(t)$ is continuous (functionally equal on the left and right of a given change point), while the governing $\eta(t)$ jumps. In principle, any $\kappa$ could be represented by infinitely many change points, but such a parameterisation is ill-posed for inference because it is infinite-dimensional. We therefore focus on the low-dimensional setting, wherein we assume a small finite number of change points. In such a regime, $\kappa(t)$ cannot be specified without the change point times $\{\gamma_j\}$.

      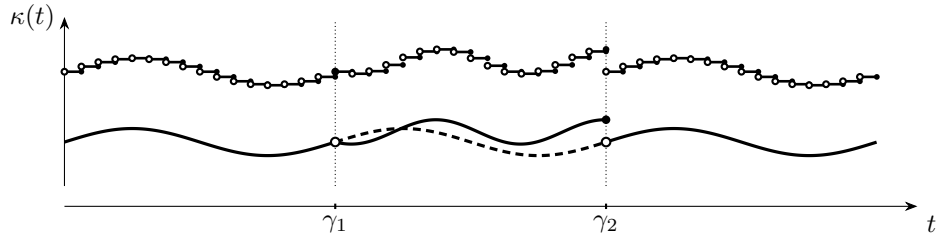
\begin{figure}[t]
    \centering
    \begin{tikzpicture}[scale=0.9,>=Stealth]
     
      \draw[->] (0,0) -- (12.6,0) node[below right] {$t$};
      \draw[->] (0,0.3) -- (0,2.8) node[left] {$\kappa(t)$};

      \def\cpA{4}
      \def\cpB{8}
      \draw[thick] (\cpA,0.04) -- (\cpA,-0.04); \node[below] at (\cpA,0) {$\gamma_1$};
      \draw[thick] (\cpB,0.04) -- (\cpB,-0.04); \node[below] at (\cpB,0) {$\gamma_2$};

      \draw[densely dotted] (\cpA,0.3) -- (\cpA,2.75);
      \draw[densely dotted] (\cpB,0.3) -- (\cpB,2.75);

      \def\c{0.95}
      \def\A{0.20}
      \def\phi{0.00}
      \def\cM{1.10}
      \def\AM{0.18}
      \def\phiM{0.35}
      \def\k{pi/2}
      \def\kM{4*pi/5}

      \draw[very thick]
        plot[domain=0:\cpA, samples=300] (\x,{\c + \A *sin(deg(\k*\x + \phi))});

      \draw[very thick,densely dashed]
        plot[domain=\cpA:\cpB, samples=300] (\x,{\c + \A *sin(deg(\k*\x + \phi))});

      \draw[very thick]
        plot[domain=\cpA:\cpB, samples=300] (\x,{\cM + \AM*sin(deg(\kM*\x + \phiM))});

      \draw[very thick]
        plot[domain=\cpB:12, samples=300] (\x,{\c + \A *sin(deg(\k*\x + \phi))});

      \pgfmathsetmacro{\ysAL}{\c  + \A *sin(deg(\k*\cpA  + \phi))}   
      \pgfmathsetmacro{\ysAM}{\cM + \AM*sin(deg(\kM*\cpA + \phiM))} 
      \pgfmathsetmacro{\ysBM}{\cM + \AM*sin(deg(\kM*\cpB + \phiM))}  
      \pgfmathsetmacro{\ysBR}{\c  + \A *sin(deg(\k*\cpB  + \phi))}   
      \fill[black]                  (\cpA,\ysAL) circle (1.8pt);    
      \fill[white,draw=black,thick] (\cpA,\ysAM) circle (1.8pt);    
      \fill[black]                  (\cpB,\ysBM) circle (1.8pt);    
      \fill[white,draw=black,thick] (\cpB,\ysBR) circle (1.8pt);

      \def\h{0.25}     
      \def\eps{0.0001}  
      \def\sep{1.04}   
    \foreach \x in {0,0.25,...,3.75}{
        \pgfmathparse{\c + \A *sin(deg(\k*(\x-\eps) + \phi)) + \sep}
        \let\y\pgfmathresult
        \draw[line width=1pt,black] (\x,\y) -- (\x+\h,\y);
        \fill[white,draw=black,thick](\x,\y) circle (1.2pt);
        \fill[black]                 (\x+\h,\y) circle (1.2pt);
      }
    
      \pgfmathsetmacro{\yPostA}{\cM + \AM*sin(deg(\kM*(\cpA+\eps) + \phiM)) + \sep}
      \draw[line width=1pt,black] (\cpA,\yPostA) -- (\cpA+\h,\yPostA);
      \fill[white,draw=black,thick](\cpA,\yPostA) circle (1.2pt);
      \fill[black]                 (\cpA+\h,\yPostA) circle (1.2pt);
  
      \foreach \x in {4.25,4.5,...,7.75}{
        \pgfmathparse{\cM + \AM*sin(deg(\kM*(\x-\eps) + \phiM)) + \sep}
        \let\y\pgfmathresult
        \draw[line width=1pt,black] (\x,\y) -- (\x+\h,\y);
        \fill[white,draw=black,thick](\x,\y) circle (1.2pt);
        \fill[black]                 (\x+\h,\y) circle (1.2pt);
      }
 
      \pgfmathsetmacro{\yPostB}{\c + \A *sin(deg(\k*(\cpB+\eps) + \phi)) + \sep}
      \draw[line width=1pt,black] (\cpB,\yPostB) -- (\cpB+\h,\yPostB);
      \fill[white,draw=black,thick](\cpB,\yPostB) circle (1.2pt);
      \fill[black]                 (\cpB+\h,\yPostB) circle (1.2pt);

      \foreach \x in {8.25,8.5,...,11.75}{
        \pgfmathparse{\c + \A *sin(deg(\k*(\x-\eps) + \phi)) + \sep}
        \let\y\pgfmathresult
        \draw[line width=1pt,black] (\x,\y) -- (\x+\h,\y);
        \fill[white,draw=black,thick](\x,\y) circle (1.2pt);
        \fill[black]                 (\x+\h,\y) circle (1.2pt);
      }
 
      \pgfmathsetmacro{\yALs}{\c  + \A *sin(deg(\k*(\cpA-\eps) + \phi)) + \sep}   
      \pgfmathsetmacro{\yARs}{\cM + \AM*sin(deg(\kM*(\cpA+\eps) + \phiM)) + \sep} 
      \fill[black]                 (\cpA,\yALs) circle (1.4pt);

      \pgfmathsetmacro{\yBLs}{\cM + \AM*sin(deg(\kM*(\cpB-\eps) + \phiM)) + \sep} 
      \pgfmathsetmacro{\yBRs}{\c  + \A *sin(deg(\k*(\cpB+\eps) + \phi)) + \sep}  
      \fill[black]                 (\cpB,\yBLs) circle (1.4pt);
    
    \end{tikzpicture}
    \caption{Variable $\kappa(t)$ with change points at $\gamma_1,\gamma_2$. Left continuous step and solid smooth functions represent $\kappa(t)$ across two change points at $\gamma_1$ and $\gamma_2$. The dashed line in $[\gamma_1,\gamma_2]$ represents a single nonstationary process with no change points.}
    \label{fig:kappa-cp-illustration}
    \end{figure}

    Our main contribution is an analytic description of the post-change regime that translates into concrete guidance for inference. Using mean field analysis, we derive explicit relaxation laws that quantify how the mean rate approaches its new level and show that information concentrates in a boundary layer immediately after the change and that information for post-change levels grows essentially linearly once beyond a short transient, whereas information for the change time itself saturates with window length. We also establish that parameters acting only on short segments are intrinsically weakly identified, which explains the ridge geometry seen between productivity level and change time. These elements motivate a Bayesian change point procedure that is both principled and numerically stable on finite windows. Mean field characterizations for Hawkes processes are well developed \cite{chevallier2017mean,clinet2017statistical,dassios2011dynamic,duarte2019stability}; our contribution is to bridge these ideas to nonstationary and change point settings and to use the resulting information geometry to design an estimation and model selection procedure.
    
    We implement a likelihood based method for offline detection of change points in the productivity function $\kappa(t)$ under a simple parametric piecewise specification while sharing background and kernel parameters across segments. This model specification stabilizes estimates, focuses attention on changes in endogenous amplification, and retains enough flexibility to capture epidemic waves arising from interventions, holidays, or behavioural shifts. The parametric form of $\kappa(t)$ is easy to specify, regularize, calibrate, and forecast, and it facilitates closed form mean field calculations for exponential kernels. We illustrate the approach with simulations and an application to contagious pneumococcal disease in The Gambia.
    
    Section~\ref{sec:model} defines the model and contrasts the variable susceptibility and variable infectivity variants. Section~\ref{sec:mean-dynamics} develops the mean dynamics and relaxation formulas. Section~\ref{sec:info} shows boundary layer localization of information, and Section~\ref{sec:mf-info} introduces a mean field surrogate for post-change information that distinguishes linear growth for levels from saturation for change time. Section~\ref{sec:ident} formalizes short regime identifiability and Section~\ref{sec:stability} studies global stability. Section~\ref{sec:estimation} builds the boundary correct likelihood and the relaxation aware Bayesian procedure, followed by simulations and the application in Section~\ref{sec:app}.

    \section{Model and standing assumptions}
    \label{sec:model}
    
    Let $N$ be a (possibly marked) simple point process on $\R_+$ with history $\mathcal H_t$. We consider the Hawkes specification
    \begin{equation}
    \label{eq:Hawkes-time-of-use}
    \lambda(t\mid \mathcal H_t)=\lambda_0  + \kappa(t;\eta_j)\sum_{t_i<t} g(t-t_i),
    \end{equation}
    where $g:\R_+\to[0,\infty)$ is measurable with $\int_0^\infty g(s) ds=1$, and $\kappa:\R_+\to[0,\infty)$ is specified piecewise on a partition $0=\gamma_0<\gamma_1<\cdots<\gamma_K=T$.  For more on Hawkes processes see \cite{daley2003introduction,hawkes1971spectra,hawkes1974cluster,reinhart2018review}.
    
    On the temporal segments $\{(\gamma_{j-1},\gamma_j]\}_{j=1}^K$ we write $\kappa(t)=\kappa_j(t;\eta_j)$ with a low-dimensional parameter $\eta_j$ (e.g.\ a level, a short ramp, or spline coefficients). In the case that $\{(\gamma_{j-1},\gamma_j]\}_{j=1}^K$ are unknown, we have a change point detection problem, as discussed in Section \ref{sec:estimation}. A change point time $\gamma$ is formally defined as a time where $\eta(\gamma-)\neq\eta(\gamma+)$.
    
    An alternative ``variable infectivity'' variant ties productivity to the parent time:
    \begin{equation}
    \label{eq:Hawkes-fixed-parent}
    \lambda(t\mid \mathcal H_t)=\lambda_0 +\sum_{t_i<t}\kappa(t_i) g(t-t_i).
    \end{equation}
    The variable susceptibility model \eqref{eq:Hawkes-time-of-use} produces mean intensity jumps at step changes in $\kappa$, whereas the variable infectivity model \eqref{eq:Hawkes-fixed-parent} yields a mean intensity that is continuous at change points, as the step in productivity is smoothed by the convolution with the kernel. The variable infectivity specification aligns with the variable productivity epidemic type aftershock (ETAS) model specification of \cite{harte2014etas}. In the context of contagious disease spread, the variable susceptibility variant reflects a temporal change in the transmission rate whereas the variable infectivity variant reflects a change in the infectivity of a newly infected.  With exponential kernels, both specifications share the same relaxation rate within subcritical segments. 
    All results in this paper extend to marked Hawkes processes, including spatiotemporal processes. In the marked case, we specify $\kappa(t,m)$, and perform change point detection over both time and mark space.
    
    \begin{assumption}[Standing conditions]
    \label{ass:standing}
    Throughout we assume:
    \begin{enumerate}[leftmargin=2em,itemsep=0.25em]
    \item[(i)] $\lambda(t)\ge\lambda_0>0$ so that logs are well defined wherever they appear.
    \item[(ii)] $\kappa(\cdot)\ge 0$ is a function of bounded variation.
    \end{enumerate}
    \end{assumption}
    Throughout, we define the log-likelihood on any finite window $[a,b]$ as \begin{equation}\label{eqn:liklihoodAB}
    \ell_{[a,b]}(\theta)
    =\int_a^b \log \lambda(t;\theta) dN(t)-\int_a^b \lambda(t;\theta) dt,
    \end{equation} which is well defined for any $\theta=(\lambda_0,\{\eta\}_{j=1}^K, \{\gamma_j\}_{j=1}^k,\theta_g)\subseteq \Theta$ where $\theta_g$ is the kernel parameters, and $\Theta$ compact.
    
    \begin{proposition}[Predictable intensity and Campbell formula]
    \label{ass:campbell}
    The (possibly marked) point process $N$ on $\mathbb R_+$ is adapted to $\{\mathcal H_t\}$ and admits an $\{\mathcal H_t\}$  predictable intensity $\lambda(t)$ such that, for every $T<\infty$,
    \[
    \int_0^T \lambda(t) dt<\infty\quad\text{a.s.}.
    \]
    Then for any non negative or integrable  $\{\mathcal H_t\}$  predictable function $f(t)$ and any $T<\infty$,
    \[
    \E \left[\int_{(0,T]} f(t) N(dt)\right]
    =
    \E \left[\int_0^T  f(t) \lambda(t)dt\right]
    \]
    see \cite[Proposition 14.2.1]{daley2007introduction} for the proof and 
    \cite{kresin2023parametric,van2010spatial} for technical details. The same identity holds conditionally given $\mathcal H_a$ whenever $f$ is supported on $(a,T]$.
    \end{proposition}
    
    In Hawkes modelling, time variation is commonly introduced through the background rate while keeping the productivity constant, e.g. \cite{omi2017hawkes}. This allows for spatially dynamic population covariates to be modelled explicitly, and is often necessary to avoid model misspecification. Our focus is complementary: we allow the productivity $\kappa(t)$ to vary while the background rate is constant. Because $\kappa(t)$ governs the reproduction mechanism (offspring per event), time variation in $\kappa$ captures changes in endogenous amplification rather than exogenous arrival pressure. Further, a step in $\kappa$ produces an instantaneous jump in the mean intensity followed by relaxation, whereas a step in the background leaves the mean continuous. Proposition~\ref{prop:no-reparam} shows there is no path wise re-parameterisation that trades time-variation in $\kappa$ for a deterministic background, so the specification brings genuinely different identification and inference behaviour.
    
    To make these distinctions precise we introduce three recurring quantities:
    \begin{align*}
    Z(t)=&\int_{-\infty}^t g(t-s) dN(s)\quad&\text{(path-wise filtered sum)}\\
    \bar\lambda(t)=&\E[\lambda(t)]\quad&\text{(mean intensity)}\\
    M(t)=&\int_{-\infty}^t g(t-s) \bar\lambda(s) ds\quad&\text{(mean-field filtered intensity)}.
    \end{align*}
    These objects are necessary for a mean-field description of the process described in Equation~\ref{eq:Hawkes-time-of-use}. We proceed by demonstrating that the conditional intensity in Equation~\ref{eq:Hawkes-time-of-use} cannot be reparameterized as a nonstationary background rate process.
    
    \begin{lemma}[Non-degeneracy of $Z(t)$]
    \label{lem:Z-nondegenerate}
    Fix $t>0$ and suppose $\lambda_0>0$. Then $Z(t)$ is non-degenerate: there exists $\varepsilon>0$ such that
    \[
    \PP \left(Z(t)-  \int_{(-\infty,t-\varepsilon]}g(t-s) dN(s)=0\right)\in(0,1).
    \]
    Consequently $\Var[Z(t)]>0$.
    \end{lemma}
    \begin{proof}
    Choose $\varepsilon>0$ with $\int_{(0,\varepsilon]} g(s) ds>0$ (this exists since $g\ge 0$ and $\int_0^\infty g=1$). Let $A$ be the event that there are no points in $(t-\varepsilon,t)$. By local integrability of $\lambda$ on $(t-\varepsilon,t)$,
    \[
    \PP \big(A\mid\mathcal H_{t-\varepsilon}\big)=\exp \left\{-\int_{t-\varepsilon}^t \lambda(u) du\right\}\in(0,1)\quad\text{a.s.},
    \]
    hence $\PP(A)\in(0,1)$. On $A$ we have $\int_{(t-\varepsilon,t)} g(t-s) dN(s)=0$, so
    \(
    Z(t)=\int_{(-\infty,t-\varepsilon]}g(t-s) dN(s)
    \)
    holds with positive probability. On $A^c$ there is at least one point in $(t-\varepsilon,t)$; since $g(t-s)>0$ on a subset of $(t-\varepsilon,t)$ of positive Lebesgue measure and $\lambda\ge \lambda_0>0$, the probability that a point falls in this subset is strictly positive, so the increment $\int_{(t-\varepsilon,t)} g(t-s) dN(s)$ is strictly positive with positive probability. Therefore the displayed event has probability strictly between $0$ and $1$, and $Z(t)$ is non-degenerate. In particular, $\Var[Z(t)]>0$.
    \end{proof}
    
    \begin{proposition}[No path-wise re-parameterisation]
    \label{prop:no-reparam}
    Consider \eqref{eq:Hawkes-time-of-use}. If $\kappa(\cdot)$ is nonconstant on any interval, there does not exist a deterministic function $\tilde\lambda_0(t)$ and a constant $\check\kappa$ such that
    \[
    \lambda(t)\equiv \tilde\lambda_0(t)+\check\kappa Z(t)\qquad\text{almost surely for all }t.
    \]
    \end{proposition}
    \begin{proof}
    Assume such $\tilde\lambda_0$ and $\check\kappa$ exist. Then
    \(
    \tilde\lambda_0(t)=\lambda_0+(\kappa(t)-\check\kappa)Z(t).
    \)
    By Lemma~\ref{lem:Z-nondegenerate}, $Z(t)$ is a random $\mathcal H_t$-measurable variable. Thus the right-hand side is random unless $\kappa(t)=\check\kappa$, which contradicts the non-constancy of $\kappa$. Therefore no such path-wise re-parameterisation exists.
    \end{proof}
    
    From a practical standpoint, holding the background rate and kernel parameters fixed while allowing the productivity $\kappa(t)$ to vary is often a natural way to represent real systems. In infectious-disease applications, for example, $\lambda_0$ and $g$ encode relatively stable features such as population covariates and disease properties, whereas $\kappa(t)$ captures changes in endogenous amplification driven by behaviour, policy, or transient contact patterns (e.g., holidays, mass gatherings), including brief explosive periods. Sharing the background and kernel parameters across segments is also statistically stabilizing: it reduces the number of free parameters, mitigates confounding between the baseline and productivity, and allows the kernel scale (e.g.\ $\beta$) to be estimated from the entire record rather than from short local windows. This pooling sharpens curvature for shared parameters, improves mixing in posterior computation, and prevents the kernel from spuriously absorbing short-window variation in $\kappa$.
    
    The sections that follow use mean dynamics to quantify post-change point relaxation, introduce a boundary-correct segment-conditional likelihood that respects carry-over, establish information localization near boundaries, develop a deterministic surrogate for post-change information growth, and derive identifiability limits for short regimes. These components together inform a Bayesian change point procedure that is both principled and stable.
    
    \section{Mean dynamics and relaxation}
    \label{sec:mean-dynamics}
    
    Mean dynamics govern how quickly the system forgets the past and how fast the post-change point mean approaches its new steady state. Such dynamics are explored in detail for stationary Hawkes processes, see \cite{bremaud2002rate,gao2018limit} for a robust description of Hawkes processes with an exponential kernel, field master equation for general kernels \cite{kanazawa2020field} and fractional Hawkes processes \cite{hainaut2020fractional}, shot noise characterizations \cite{dassios2011dynamic}, and multivariate Hawkes settings with asymptotic results \cite{chevallier2017mean,clinet2017statistical,duarte2019stability}. 
    
    Our study of mean field dynamics in the context of a varying productivity Hawkes model allows a practitioner to determine how long a post-change window must be to accumulate meaningful evidence for a level change. At the start of any finite observation window there is, by default, a boundary: unless the process is truly observed from inception (or we condition on the infinite past), pre-window events enter as an unknown carry-over that acts like an initial change. The formulas below therefore treat the left edge of the window as a boundary layer and quantify the associated relaxation. These results serve both forward modelling and calibration, and they feed into the information calculations, identifiability limits, and prior choices used later.
    
    Taking expectations in \eqref{eq:Hawkes-time-of-use} and using $\E[dN(t)\mid\mathcal H_t]=\lambda(t) dt$, given Proposition \ref{ass:campbell}, yields the linear Volterra system
    \begin{equation*}
    \label{eq:Volterra}
    \bar\lambda(t)=\lambda_0+\kappa(t) M(t),
    \qquad
    M(t)=\int_{-\infty}^t g(t-s) \bar\lambda(s)\, ds
    \end{equation*}
    or
    \begin{equation*}
    \bar\lambda(t)=\lambda_0+\kappa(t) \int_{-\infty}^t g(t-s) \bar\lambda(s)\, ds,
    \qquad
    M(t)=\lambda_0+\int_{-\infty}^t g(t-s)\kappa(s) M(s)\, ds.
    \end{equation*}
    Similar mean field representations of Hawkes processes are derived in detail in \cite{chevallier2017mean,dassios2011dynamic}. Both equations can be solved efficiently using standard numerical techniques. In the case where
    \[g(s)=\beta e^{-\beta s}\1_{\{s>0\}}\] 
    is an exponential kernel we can solve these equations explicitly, making further analysis easily tractable. In this case differentiating $M$ gives
    \(
    M'(t)=-\beta M(t)+\beta \bar\lambda(t)
    \),
    so
    \begin{equation}
    \label{eq:M-ode}
    M'(t)=\beta\lambda_0+\beta(\kappa(t)-1)M(t),\qquad
    \bar\lambda(t)=\lambda_0+\kappa(t)M(t).
    \end{equation}
    Solving \eqref{eq:M-ode} by variation of constants yields, for $t\ge t_0$,
    \begin{equation}
        \label{eq:odesol}
        M(t) =
    e^{-\beta \int_{t_0}^t (1-\kappa(s) )\, ds}
    M(t_0) + \beta \lambda_0
    \int_{t_0}^t
    e^{-\beta \int_{u}^t (1-\kappa(r))\, dr}
    du.
    \end{equation}
    The first term describes exponential forgetting of the initial condition, at a rate governed by the path-average of $1-\kappa(s)$; the second term describes the pull towards the current reproduction regime.
    
    In particular, if $\kappa\neq 1$ is constant on an interval $t>t_0$, then
    \begin{equation}
    \label{eq:M-const-solution}
    M(t)=\frac{\lambda_0}{1-\kappa}
    +\Big(M(t_0^+)-\frac{\lambda_0}{1-\kappa}\Big)e^{-\beta(1-\kappa)(t-t_0)},\quad
    \bar\lambda(t)=\frac{\lambda_0}{1-\kappa}
    +\Big(\bar\lambda(t_0^+)-\frac{\lambda_0}{1-\kappa}\Big)e^{-\beta(1-\kappa)(t-t_0)},
    \end{equation}
    whereas in the critical case $\kappa=1$ we obtain
    \[
    M(t)=M(t_0^+)+\beta\lambda_0(t-t_0),\qquad
    \bar\lambda(t)=\bar\lambda(t_0^+)+\beta\lambda_0(t-t_0).
    \]
    For $\kappa<1$, \eqref{eq:M-const-solution} shows that $M$ and $\bar\lambda$ relax exponentially towards the steady state $\Lambda(\kappa):=\lambda_0/(1-\kappa)$, with relaxation time
    \[
    \tau=\frac{1}{\beta(1-\kappa)}.
    \]
    $\tau$ is the time it takes for the transient term to shrink by a factor $e$. Near-critical segments ($\kappa\approx 1$) have long relaxation times and therefore remember the pre-change regime much longer than subcritical segments. Practically, \(\tau\) is the time scale on which the system ``forgets'' a change. After about one \(\tau\), the mean has moved most of the way towards the new steady state; after a few \(\tau\)'s, the influence of the old regime is negligible. Near-critical regimes (\(\kappa\approx 1\)) have very large \(\tau\), so they carry a long memory of what happened before the change. In epidemic terms, a near-critical phase behaves like a system that keeps ``echoing'' past transmission conditions long after an intervention, whereas a strongly subcritical phase forgets quickly. 
    
    \begin{remark}[Variable infectivity mean dynamics]
    \label{rem:fixed-parent-mf}
    For the variable infectivity specification \eqref{eq:Hawkes-fixed-parent} with exponential kernel $g(s)=\beta e^{-\beta s}\mathbbm{1}_{\{s>0\}}$ it is convenient to define
    \[
    Z_{\mathrm{fix}}(t)
    =\sum_{t_i<t}\kappa(t_i)\,g(t-t_i),\qquad
    M_{\mathrm{fix}}(t)=\E[Z_{\mathrm{fix}}(t)],\qquad
    \bar\lambda_{\mathrm{fix}}(t)=\lambda_0+M_{\mathrm{fix}}(t).
    \]
    Then $Z_{\mathrm{fix}}$ satisfies
    \[
    dZ_{\mathrm{fix}}(t)
    =-\beta Z_{\mathrm{fix}}(t)\,dt+\beta\kappa(t)\,dN(t),
    \]
    so taking expectations gives
    \[
    M_{\mathrm{fix}}'(t)
    =-\beta M_{\mathrm{fix}}(t)+\beta\kappa(t)\bar\lambda_{\mathrm{fix}}(t)
    =\beta\kappa(t)\lambda_0+\beta\{\kappa(t)-1\}M_{\mathrm{fix}}(t),
    \qquad
    \bar\lambda_{\mathrm{fix}}(t)=\lambda_0+M_{\mathrm{fix}}(t).
    \]
    On any interval where $\kappa(t)\equiv\kappa\in[0,1)$ is constant this reduces to
    \[
    M_{\mathrm{fix}}'(t)=\beta\kappa\lambda_0+\beta(\kappa-1)M_{\mathrm{fix}}(t),
    \]
    with solution
    \[
    M_{\mathrm{fix}}(t)
    =\frac{\kappa\lambda_0}{1-\kappa}
    +\Big(M_{\mathrm{fix}}(t_0^+)-\frac{\kappa\lambda_0}{1-\kappa}\Big)
    e^{-\beta(1-\kappa)(t-t_0)}
    \]
    and hence
    \[
    \bar\lambda_{\mathrm{fix}}(t)
    =\frac{\lambda_0}{1-\kappa}
    +\Big(\bar\lambda_{\mathrm{fix}}(t_0^+)-\frac{\lambda_0}{1-\kappa}\Big)
    e^{-\beta(1-\kappa)(t-t_0)}.
    \]
    Thus on each constant-productivity segment the variable susceptibility intensity \eqref{eq:Hawkes-time-of-use} and the variable infectivity intensity \eqref{eq:Hawkes-fixed-parent} share the same stationary mean $\Lambda(\kappa)=\lambda_0/(1-\kappa)$ and relaxation time $\tau=1/[\beta(1-\kappa)]$. The only qualitative difference is that the variable susceptibility mean intensity jumps at a step in $\kappa$, whereas the variable infectivity mean is continuous and then relaxes with the same rate. Consequently, the mean-field information results that depend only on $\Lambda(\kappa)$ and $\tau$ (boundary-layer width, linear growth of post-change information for levels, saturation and locality of change-time information) carry over to the variable infectivity model. Because there is no instantaneous jump under the variable infectivity dynamics, the variable susceptibility change-time calculations should be interpreted as an optimistic benchmark, in the mean-field approximation, for how sharply a change point can be localised under the variable infectivity specification.
    \end{remark}

    The solution \eqref{eq:odesol} also makes clear that, when $\kappa(t)$ remains uniformly below one, the influence of any initial condition $M(t_0)$ decays at least as fast as 
    \[
    \exp\Bigl\{-\beta\inf_{[t_0,t]}(1-\kappa(s))(t-t_0)\Bigr\},
    \]
    so that past events are effectively forgotten after a few local relaxation times.
    
    When $\kappa(t)$ varies slowly relative to the relaxation time, $M(t)$ tracks the ``instantaneous'' fixed point $\Lambda(\kappa(t))=\lambda_0/(1-\kappa(t))$ up to a small lag. Formally, differentiating $\Lambda(\kappa(t))$ and comparing with \eqref{eq:M-ode} suggests the dimensionless ``slow-variation'' condition
    \[
    \frac{|\kappa'(t)|}{\beta(1-\kappa(t))^2}\ll 1,
    \]
    under which the relative deviation $|\bar\lambda(t)-\Lambda(\kappa(t))|/\Lambda(\kappa(t))$ remains small. This explains why near-critical segments (small $1-\kappa$) both relax slowly and exhibit larger tracking error when $\kappa$ is allowed to change rapidly. Operationally, these dynamics imply that the quality of any estimator depends on the relaxation timescale. Specifically, this motivates using post-change windows of length at least a few $\tau$ to ensure the log-likelihood accumulates sufficient curvature for identification.
    
    When $\kappa$ jumps from $\kappa_1$ to $\kappa_2$ at $t^\ast$ and the pre-change mean is in steady state, we have $M(t^\ast)=\Lambda(\kappa_1)=:\Lambda_1$ and hence
    \[
    \bar\lambda(t^{\ast-})=\Lambda_1,\qquad
    \bar\lambda(t^{\ast+})=\lambda_0+\kappa_2\Lambda_1.
    \]
    Subtracting the post-change steady state $\Lambda_2:=\Lambda(\kappa_2)$ gives
    \begin{equation}
    \label{eq:jump-offset}
    \bar\lambda(t^{\ast+})-\Lambda_2
    =\lambda_0 \frac{\kappa_2(\kappa_1-\kappa_2)}{(1-\kappa_1)(1-\kappa_2)},
    \end{equation}
    followed by exponential relaxation at rate $\beta(1-\kappa_2)$. Thus a step in $\kappa$ produces an instantaneous jump in the mean intensity, whose magnitude depends on both the pre- and post-change regimes, and then an exponential smoothing governed by the post-change relaxation time. These features are the core of the information localization and surrogate information calculations that follow.

    \subsection{Information localisation near changes}
    \label{sec:info}
    
    Change point detection succeeds when information about local parameters concentrates near the change point. Our goal is to understand where along the time axis the data are actually informative about post-change parameters. The answer is sharply localized: for parameters that act only through the post-change productivity $\kappa_j$ on a segment $(\gamma_{j-1},\gamma_j]$, the Fisher information density is largest immediately after $\gamma_{j-1}$ and then decays on the kernel time scale. (Equivalently, if $\kappa_j$ is not constant, this applies to $\eta_j$). In other words, most of the information about a new regime is generated in the first few relaxation times after the change; observations much later in the same segment behave more like steady-state data. 
    Thus short post-change windows can suffice for detection, and early events dominate likelihood ratios between nearby partitions. Moreover, data collected past a few relaxation times will not contribute much towards estimating $\gamma_{j-1}$.
    
    Fix a segment $(\gamma_{j-1},\gamma_j]$ and decompose $Z(t)=B_j(t)+W_j(t)$ into a carry-over term from the pre-segment history and a within-segment contribution. For $g(s)=\beta e^{-\beta s}\1_{\{s>0\}}$ and $t>\gamma_{j-1}$,
    \[
    B_j(t)=e^{-\beta(t-\gamma_{j-1})}\sum_{t_i<\gamma_{j-1}}\beta e^{-\beta(\gamma_{j-1}-t_i)}
    =:S_j(\beta) e^{-\beta(t-\gamma_{j-1})}.
    \]
    For a scalar parameter $\eta$ entering only via $\kappa_j(t;\eta)$ on $(\gamma_{j-1},\gamma_j]$, define the information functional and its expectation on $[a,b]\subseteq(\gamma_{j-1},\gamma_j]$ by
    \[
    \mathcal J_\eta([a,b])=\int_a^b \frac{\{\partial_\eta \lambda(t)\}^2}{\lambda(t)} dt,
    \qquad
    \mathcal I_\eta([a,b])=\E[\mathcal J_\eta([a,b])].
    \]
    Since $\partial_\eta\lambda(t)=\big(\partial_\eta\kappa_j(t;\eta)\big)Z(t)$, two mechanisms concentrate information near $\gamma_{j-1}$. First, immediately after $\gamma_{j-1}$ the quantity $Z(t)=B_j(t)+W_j(t)$ is dominated by $B_j(t)$, which decays at the kernel rate; hence
    \[
    Z(t)\ \ge\ B_j(t)=S_j(\beta) e^{-\beta (t-\gamma_{j-1})},\qquad t\downarrow \gamma_{j-1}.
    \]
    With $\kappa_{\max}:=\sup_{(\gamma_{j-1},\gamma_j]}\kappa_j(t;\eta)$ and for a step-height parameter ($\partial_\eta\kappa_j\equiv 1$) we obtain the explicit lower bound
    \begin{align}
    \label{eq:info-near}
    \mathcal J_\eta([\gamma_{j-1},\gamma_{j-1}+\delta])
    \ \ge&\
    \int_0^\delta 
    \frac{S_j(\beta)^2 e^{-2\beta u}}
    {\lambda_0+\kappa_{\max} S_j(\beta) e^{-\beta u}} du\\
    =&\frac{1}{\beta \kappa_{\max}^2} \left[
    \kappa_{\max} S_j(\beta)(1-e^{-\beta\delta})
    -\lambda_0\log \frac{\lambda_0+\kappa_{\max} S_j(\beta)}{\lambda_0+\kappa_{\max} S_j(\beta)e^{-\beta\delta}}
    \right] .
    \end{align}
    The assumption $\partial_\eta\kappa_j\equiv 1$ is made only to obtain the closed-form bound \eqref{eq:info-near} for a canonical step-height (level) parameterization; for general parametrizations the same localization bound holds up to the multiplicative factor $\inf_{u\in[0,\delta]}(\partial_\eta\kappa_j(\gamma_{j-1}+u;\eta))^2$.

     While the explicit bound \eqref{eq:info-near} is exact, its algebraic complexity obscures the dependence on the carry-over magnitude. By inspecting the integrand as $\delta \downarrow 0$, we recover the instantaneous information density at the boundary. The integrand is continuous in $u$ and hence for small $\delta$, 
    \begin{equation}
    \label{eq:info-short-delta}
    \mathcal J_\eta([\gamma_{j-1},\gamma_{j-1}+\delta])
    \ \ge\
    \delta\,\frac{S_j(\beta)^2}{\lambda_0+\kappa_{\max} S_j(\beta)}\ +\ O(\delta^2).
    \end{equation}
    This expansion isolates the effective signal-to-noise ratio immediately following a change: detection power scales quadratically with the carry-over $S_j(\beta)$ when the background $\lambda_0$ dominates, but only linearly when the carry-over itself dominates the intensity. Note that if $\partial_\eta\kappa_j(t;\eta)$  places weight near $\gamma_{j-1}$ (e.g.\ a smoothed step or smoothed change time), then the large carry-over term in $Z(t)$ further concentrates the score in the boundary layer. For ramp slopes with $\partial_\eta\kappa_j\propto (t-\gamma_{j-1})$, this vanishing at the boundary offsets the effect and shifts the peak to lag $O(1/\beta)$ while remaining kernel-localised.
    
    \subsection{Post-change mean-field surrogate information}
    \label{sec:mf-info}
    
    For design, calibration, and priors it is helpful to have a simple deterministic surrogate that captures the rate at which information grows after a change. The mean-field surrogate derived here plays that role, yielding a closed form for level changes and an explicit saturation bound for change-time information. The results in this section assume $\kappa<1$. Supercritical regimes with unknown change times are discussed in  Subsection~\ref{subsec:explosive}. In this subsection, $\vartheta$ denotes a generic scalar parameter (e.g.\ $\kappa_2$ or $t^\ast$), i.e. a coordinate or component of the full parameter vector~$\theta$.

    Unlike $\mathcal I_\eta([a,b])=\E[\mathcal J_\eta([a,b])]$, which generally lacks a closed form, we now introduce a deterministic mean-field surrogate $\mathcal I_\vartheta^{\mathrm{mf}}$ obtained by replacing $Z$ with $M$ and $\lambda$ with $\bar\lambda$. In particular,
   \[
\mathcal I_\vartheta^{\mathrm{mf}}([0,\Delta])
:=\int_0^\Delta \frac{\big(S_\vartheta^{\mathrm{mf}}(u)\big)^2}{\bar\lambda(u)}\,du,
\]
    where $S_\vartheta^{\mathrm{mf}}(u)$ denotes the mean-field surrogate of the path-wise score factor $\partial_\vartheta\lambda(u)$ obtained by replacing $Z$ with $M$ and $\lambda$ with $\bar\lambda$. To keep the mean-field information tractable, we specialize to the canonical single-step model, i.e. $\kappa(t)$ is constant within each regime and jumps at $t^\ast$. This is the natural case for change-point inference, and it also serves as a local approximation when $\kappa(t)$ varies slowly relative to the relaxation time. Consider a single step at $t^\ast$,
    \[
    \kappa(t)=\kappa_1 \1_{\{t<t^\ast\}}+\kappa_2 \1_{\{t\ge t^\ast\}},\qquad
    0\le \kappa_1,\kappa_2<1,
    \]
    with exponential kernel. On $u:=t-t^\ast\ge 0$, the mean-field dynamics are
    \begin{equation}
    \label{eq:mf-ode-post}
    M'(u)=\beta\lambda_0+\beta(\kappa_2-1)M(u),\qquad
    \bar\lambda(u)=\lambda_0+\kappa_2 M(u).
    \end{equation}
    Let
    \[
    \varrho:=\beta(1-\kappa_2),\qquad
    \Lambda_i:=\frac{\lambda_0}{1-\kappa_i}\ (i=1,2),\qquad
    A:=\Lambda_1-\Lambda_2.
    \]
    Solving \eqref{eq:mf-ode-post} with $M(0)=\Lambda_1$ yields
    \begin{equation}
    \label{eq:M-post}
    M(u)=\Lambda_2+A e^{-\varrho u},\qquad
    \bar\lambda(u)=\Lambda_2+\kappa_2 A e^{-\varrho u}.
    \end{equation}
    In the mean-field surrogate for $\kappa_2$, we replace the path-wise factor $Z(u)=\partial_{\kappa_2}\lambda(u)$ by $M(u)$ and set
    \[
    S_{\kappa_2}^{\mathrm{mf}}(u):=M(u),\qquad
    \mathcal I_{\kappa_2}^{\mathrm{mf}}(\Delta)=\int_0^\Delta \frac{M(u)^2}{\bar\lambda(u)} du.
    \]
    It follows that 
    \begin{align}
    \label{eq:Ikappa2-mf}
    \mathcal I_{\kappa_2}^{\mathrm{mf}}(\Delta)
    =\int_0^\Delta \frac{M(u)^2}{\bar\lambda(u)} du
    =\int_0^\Delta \frac{(\Lambda_2+A e^{-\varrho u})^2}{\Lambda_2+\kappa_2 A e^{-\varrho u}} du.
    \end{align}
    With $y=e^{-\varrho u}$, $du=-(1/\varrho)y^{-1}dy$, one obtains the closed form
    \begin{equation}
    \label{eq:Ikappa2-mf-closed}
    \mathcal I_{\kappa_2}^{\mathrm{mf}}(\Delta)
    =\frac{1}{\varrho \kappa_2^2}\Big[
    A\kappa_2\big(1-e^{-\varrho \Delta}\big)
    +\Lambda_2 \kappa_2^2 \varrho\Delta
    + \Lambda_2(1-\kappa_2)^2 
    \log \frac{A\kappa_2 e^{-\varrho \Delta}+\Lambda_2}{A\kappa_2+\Lambda_2}
    \Big].
    \end{equation}
    
    To interpret the closed form, note that it arises by replacing the random filtered sum and intensity with their means: $Z(u)\mapsto M(u)=\E[Z(u)]$ and $\lambda(u)\mapsto\bar\lambda(u)=\E[\lambda(u)]=\lambda_0+\kappa_2 M(u)$. Let $\mathcal I_{\kappa_2}^{\text{true}}(\Delta):=\int_0^\Delta \E \left[\frac{Z(u)^2}{\lambda_0+\kappa_2 Z(u)}\right]du$ be the exact (expected) post-change Fisher information when the change time is known. Writing $f(z)=z^2/(\lambda_0+\kappa_2 z)$, we have $f''(z)=2\lambda_0^2/(\lambda_0+\kappa_2 z)^3>0$, so $f$ is strictly convex. By Jensen's inequality,
    \[
    \mathcal I_{\kappa_2}^{\text{true}}(\Delta)=\int_0^\Delta \E[f(Z(u))] du
    \ \ge\ \int_0^\Delta f(\E[Z(u)]) du
    =\mathcal I_{\kappa_2}^{\mathrm{mf}}(\Delta),
    \]
    with equality only if $\Var(Z(u))\equiv0$. A second-order delta-method expansion around $M(u)$ makes the (positive) Jensen gap explicit:
    \[
    \mathcal I_{\kappa_2}^{\text{true}}(\Delta)-\mathcal I_{\kappa_2}^{\mathrm{mf}}(\Delta)
    \ \approx\ \int_{0}^{\Delta}\frac{\lambda_0^2 \Var(Z(u))}{\big(\lambda_0+\kappa_2 M(u)\big)^3} du.
    \]
    After a down-step ($\kappa_1>\kappa_2$), the carry-over variability inflates $\Var(Z(u))$ near the boundary, so the surrogate underestimates most strongly. A larger baseline $\lambda_0$ both reduces the relative variability of $Z$ (more independent immigrant clusters) and decreases $f''(M(u))$, shrinking the gap. 
    
    This scaling behaviour highlights a critical distinction between information magnitude and temporal geometry. Increasing the background rate $\lambda_0$ scales the Fisher information roughly linearly (improving global detectability) and suppresses the relative impact of the Jensen gap, making the mean-field surrogate increasingly accurate. However, the relaxation time $\tau=[\beta(1-\kappa)]^{-1}$ is independent of $\lambda_0$ (for fixed $\beta,\kappa$). For change-time inference, the information density is concentrated in the transient of width $O(\tau)$ and saturates once $u\gg\tau$; a larger $\tau$ therefore spreads the onset contrast over a longer interval, making the profile likelihood in $t^\ast$ flatter on any fixed post-change window and strengthening the need for smoothing/regularisation when regimes are short relative to $\tau$. 
    
    More formally, as a function of the post-change window length $\Delta$, the mean-field information has the large-$\Delta$ expansion
    \begin{equation}
    \label{eq:Ikappa2-mf-asymp}
    \mathcal I_{\kappa_2}^{\mathrm{mf}}(\Delta)
    =\Lambda_2 \Delta + C(\lambda_0,\beta,\kappa_1,\kappa_2) + o(1).
    \end{equation}
    Thus, beyond a short transient, curvature for the post-change level grows essentially linearly with slope $\Lambda_2$.
    
    Having characterized post-change information for the level $\kappa_2$, we now turn to the information for the change time $t^\ast$. The likelihood score for a change point contains two components: a singular contribution from the instantaneous jump in intensity at $t^\ast$, and a regular contribution from the subsequent relaxation of the mean intensity.
    
    The singular ``jump information'' arises because the intensity creates an instantaneous contrast
    \[
    \lambda_{\mathrm{jump}}
    \;:=\;
    \bar\lambda(t^\ast+)-\bar\lambda(t^\ast-)
    =
    (\kappa_2-\kappa_1)\Lambda_1,
    \]
    when the pre-change regime is in steady state, so that $\bar\lambda(t^\ast-)=\Lambda_1$ and $\bar\lambda(t^\ast+)=\lambda_0+\kappa_2\Lambda_1$.  In the mean-field limit, this contributes a constant initial information that can be approximated by the squared signal-to-noise ratio of the jump:
    \begin{equation}
    \label{eq:I-jump}
    \mathcal I_{\mathrm{jump}}
    \approx
    \frac{\lambda_{\mathrm{jump}}^{\,2}}{\bar\lambda(t^\ast+)}
    =
    \frac{(\kappa_2-\kappa_1)^2\Lambda_1^2}{\lambda_0+\kappa_2\Lambda_1}.
    \end{equation}
    
    The relaxation component tracks how the mean intensity evolves after the jump. Writing $u=t-t^\ast$, the linearized mean-field dynamics yield
    \[
    \{\partial_{t^\ast}\bar\lambda(u)\}_{\text{mf}}
    =\kappa_2\,\partial_{t^\ast}M(u)
    =Be^{-\varrho u}\1_{\{u\ge 0\}},
    \qquad
    B:=-\beta\kappa_2(\kappa_2-\kappa_1)\Lambda_1,
    \]
    with $\varrho=\beta(1-\kappa_2)$ as before. Combining the singular and regular pieces, the total mean-field information for $t^\ast$ over a post-change window of length $\Delta$ is
    \begin{equation}
    \label{eq:Itstar-mf}
    \mathcal I_{t^\ast}^{\mathrm{mf}}(\Delta)
    = \mathcal{I}_{\text{jump}} + \int_0^\Delta \frac{B^2 e^{-2\varrho u}}{\bar\lambda(u)} du.
    \end{equation}
    If $\underline\lambda\le \bar\lambda(u)$ for all $u\ge 0$ (e.g.\ $\underline\lambda=\min(\Lambda_1,\Lambda_2)$), then
    \begin{equation}
    \label{eq:Itstar-mf-bound}
    \mathcal I_{t^\ast}^{\mathrm{mf}}(\Delta)
    \ \le\ 
    \mathcal{I}_{\text{jump}} + \frac{B^2}{\underline\lambda}\int_0^\Delta e^{-2\varrho u} du
    \ =\
    \mathcal{I}_{\text{jump}} + \frac{B^2}{2\varrho \underline\lambda}(1-e^{-2\varrho\Delta}),
    \end{equation}
    so the relaxation contribution converges rapidly and the total information saturates in $\Delta$. The change-time parameter is therefore strictly local: its information depends on the jump size $\lambda_{\mathrm{jump}}$ and the immediate relaxation path, but does not accumulate indefinitely as the post-change window grows. This saturation aligns with general asymptotic results for change points in point processes, which establish that estimation error is governed by the local jump magnitude rather than the observation horizon \cite{ibragimov2013statistical, dachian2008goodness}. However, \eqref{eq:Itstar-mf-bound} provides an explicit, tractable bound for this limit specifically governed by the Hawkes relaxation dynamics.
    
    From a practical standpoint, this locality also explains why the observed Hessian with respect to $t^\ast$ can be highly unstable on a single realisation, and why there is a systematic gap between the mean-field curvature and the precision of any path-wise estimator. As a function of $t^\ast$ the log-likelihood $\ell(t^\ast)$ is piecewise smooth with jumps whenever the proposed change point crosses an event time: small shifts in $t^\ast$ move individual events across the boundary, causing discontinuous changes in the score and producing a jagged landscape of local optima (reference Figure \ref{fig:sim_profiles}). By contrast, $\mathcal I_{t^\ast}^{\mathrm{mf}}$ is defined in terms of the curvature of the expected intensity path $\bar\lambda(t)$, as if the conditional intensity were observed continuously. In this idealized limit the standard local asymptotic normality argument identifies the Fisher information with the curvature of the log-likelihood, leading to the saturation bound \eqref{eq:Itstar-mf-bound}. For the realisation, however, we only observe discrete arrivals; the jaggedness of the path-wise likelihood effectively limits localisation to the scale of the inter-arrival times, and the empirical precision $1/\Var(\hat t^\ast)$ saturates below the mean-field limit. In this sense $\mathcal I_{t^\ast}^{\mathrm{mf}}$ is best interpreted as the information content of the macroscopic envelope. It quantifies the width and height of the posterior ridge one would obtain if the intensity were continuously observed, whereas any single-path curvature estimate probes a discretised, noisy version of this envelope.
    
   A Bayesian formulation mitigates this instability by integrating over the change time rather than relying on a local quadratic approximation of the single-path log-likelihood. The posterior for $t^\ast$ is proportional to $\exp\{\ell(t^\ast)\}\pi(t^\ast)$ and therefore averages over the jagged event-induced microstructure by spreading mass across nearby candidate change times, yielding credible sets that reflect both saturation and locality. In our framework, the mean-field information plays a calibrating role (it quantifies how local the problem is and the best-case sharpness one could hope for), while uncertainty quantification for change times is based on the full posterior rather than on a Gaussian approximation built from a single-path Hessian.
    
    \subsubsection{Fokker-Planck refinement}
    
    Fokker-Planck and master equations are well developed in the case of Hawkes processes \cite{hainaut2020fractional,kanazawa2020field}. Because of the exponential kernel, \(Z_t\) is a piecewise deterministic Markov process \cite{dion2020exponential} with
    \(dZ_t=-\beta Z_t\,dt+\beta\,dN_t\) and rate \(\lambda_t=\lambda_0+\kappa(t)Z_t\).
    Let \(p(t,z)\) be the density of \(Z_t\) on \(z\ge0\). The Fokker-Planck equation is
    \begin{equation}
    \label{eq:master-fwd}
    \partial_t p(t,z)
    = \beta\,\partial_z\!\big[z\,p(t,z)\big]
    +\big(\lambda_0+\kappa(t)\,(z-\beta)\big)\,p(t,z-\beta)
    -\big(\lambda_0+\kappa(t)\,z\big)\,p(t,z),
    \end{equation}
    with \(p(t,z)=0\) for \(z<0\) and initial \(p(0,\cdot)\) inherited from the pre change regime. Testing \eqref{eq:master-fwd} with \(1,z,z^2\) yields the closed moment ODEs
    \begin{align}
    \label{eq:moment-odes}
    M'(t)&=\beta\lambda_0-\beta(1-\kappa(t))M(t),\\
    V'(t)&=-2\beta(1-\kappa(t))V(t)+\beta^2\!\big(\lambda_0+\kappa(t)\,M(t)\big),\qquad V(t):=\Var[Z_t]. 
    \end{align}
    Inserting these into a second order delta expansion refines \(\mathcal I_\vartheta^{\mathrm{mf}}\). For any scalar parameter \(\vartheta\) entering only via \(\kappa(t;\vartheta)\) (i.e.\ through the \(\kappa\)-component of \(\theta\))
    \begin{equation}
    \label{eq:Ikappa-mf-plus}
    \mathcal I_\vartheta^{\mathrm{mf+}}([a,b])
    \;\approx\;\int_a^b \big(\partial_\vartheta\kappa(t)\big)^2
    \left\{\frac{M(t)^2}{\bar\lambda(t)}+\frac{\lambda_0^2\,V(t)}{\bar\lambda(t)^3}\right\}dt,
    \qquad \bar\lambda(t)=\lambda_0+\kappa(t)M(t).
    \end{equation}
    This adds a positive variance term that is largest immediately after a down step and decays at rate \(\beta(1-\kappa)\). It does not replace the closed form \eqref{eq:Ikappa2-mf-closed}; rather, \(\mathcal I^{\mathrm{mf}}\) remains the transparent baseline (and a lower bound), while \(\mathcal I^{\mathrm{mf+}}\) tightens short window calibration.
    
    \subsection{Discussion}
    
    Beyond characterising the boundary-layer transient induced by a change (and its exponential decay on the relaxation time scale), the mean dynamics determine how fast statistically usable curvature appears after a change and therefore guide both window-length selection and regularization. For a post-change level $\kappa_2$ with known change time $t^\ast$, the mean-field information density at lag $u:=t-t^\ast\ge0$ is
    \[
    \frac{d}{du}\,\mathcal I_{\kappa_2}^{\mathrm{mf}}(u)
    =\frac{M(u)^2}{\bar\lambda(u)},
    \qquad
    M(u)=\Lambda_2+A e^{-\varrho u},\quad
    \bar\lambda(u)=\lambda_0+\kappa_2 M(u),
    \]
    with $\Lambda_2=\lambda_0/(1-\kappa_2)$, $A=\Lambda_1-\Lambda_2$, and $\varrho=\beta(1-\kappa_2)$ as in \eqref{eq:mf-ode-post}-\eqref{eq:M-post}. As $u$ grows, this density converges to $\Lambda_2$, so the total curvature over $[0,\Delta]$ satisfies the linear  plus  transient decomposition
    \[
    \mathcal I_{\kappa_2}^{\mathrm{mf}}(\Delta)=\Lambda_2\,\Delta+O(1),
    \]
    exactly as stated in \eqref{eq:Ikappa2-mf-asymp}. Moreover, the convergence to the linear regime is exponentially fast. Writing $\phi(x):=\frac{(\Lambda_2+x)^2}{\Lambda_2+\kappa_2 x}$, one has $M(u)-\Lambda_2=A e^{-\varrho u}$ and $\bar\lambda(u)=\Lambda_2+\kappa_2(M(u)-\Lambda_2)$, so
    \begin{equation}
    \label{eq:info-density-dev}
    \bigg|\frac{M(u)^2}{\bar\lambda(u)}-\Lambda_2\bigg|
    =\big|\phi(M(u)-\Lambda_2)-\phi(0)\big|
    \le C\,|A|\,e^{-\varrho u},
    \end{equation}
    for a finite constant $C=C(\Lambda_2,\kappa_2,|A|)$ because $\phi$ is smooth on $[-|A|,|A|]$. Integrating \eqref{eq:info-density-dev} yields
    \begin{equation}
    \label{eq:info-transient-tail}
    \big|\ \mathcal I_{\kappa_2}^{\mathrm{mf}}(\Delta)-\Lambda_2\Delta\ -\ C_0\ \big|
    \le \frac{C\,|A|}{\varrho}\,e^{-\varrho\Delta},
    \end{equation}
    for some constant $C_0$ depending on $(\lambda_0,\beta,\kappa_1,\kappa_2)$. In particular, to ensure that the per-time curvature is within a tolerance $\varepsilon$ of its steady value, or, equivalently, that the transient contribution beyond $\Delta$ is at most $\varepsilon$, it suffices to take
    \[
    \Delta\ \gtrsim\ \frac{1}{\beta(1-\kappa_2)}\ \log\!\Big(\frac{C\,|A|}{\varepsilon}\Big)
    =\tau_2\ \log\!\Big(\frac{C\,|A|}{\varepsilon}\Big),
    \]
    that is, only a few multiples of the post-change relaxation time $\tau_2=1/[\beta(1-\kappa_2)]$. This makes precise the informal statement that one must observe for a few relaxation times after the boundary so that the information accumulation rate stabilises near its steady-state value $\Lambda_2$ and the residual boundary-layer correction becomes negligible (of order $e^{-\Delta/\tau_2}$).

    The same dynamics explain why extending the window eventually ceases to help for the change-time parameter. When $t^\ast$ is unknown, the corresponding mean-field integrand decays like $e^{-2\varrho u}$, and the bound \eqref{eq:Itstar-mf-bound} shows that the total information for $t^\ast$ saturates. Thus, while level information grows essentially linearly after a short boundary layer, change-time information exhibits sharply diminishing returns beyond a few relaxation times. This formalizes the dichotomy between local change point detection (where only a boundary layer of width $O(\tau_2)$ matters) and global level estimation (where effective information scales with the usable post-change horizon). 
    

    Finally, our mean field analysis provides concrete guidance on the observation window $\Delta$ required after a suspected change time $t^\ast$. Since information for the change time saturates once the intensity relaxes to its new steady state, extending the window beyond the transient regime yields diminishing returns for localising $t^\ast$. Specifically, with the post-change relaxation time defined as $\tau_2 = [\beta(1-\kappa_2)]^{-1}$, the mean intensity completes approximately $95\%$ of its adjustment by $3\tau_2$ and $99\%$ by $5\tau_2$.  Therefore, to maximise detection efficiency without collecting redundant steady-state data, we recommend a dynamic window of length $\Delta \approx 3\tau_2$ to $5\tau_2$. Collecting data beyond this horizon improves the precision of the level estimate $\widehat\kappa_2$ linearly, but does not materially improve the precision of the change point estimate $\widehat t^\ast$. For variable $\kappa$, given \eqref{eq:odesol}, one can expect a similar completion of the adjustment for the time window $\Delta$ when $\beta\int_{t^\ast}^{\Delta+t^\ast}(1-\kappa(s))\,ds$ is between 3 and 5.

    \section{Identifiability and short-regime limits}
    \label{sec:ident}
    
    Information localization and the surrogate growth rates together imply that short regimes are a fundamentally difficult inference problem: the total curvature available for parameters that act only on a short segment is bounded, regardless of the ambient horizon. This section formalizes that statement and explains the high posterior variance and estimator dispersion seen for short regimes in practice.
    
   Consider the variable susceptibility specification, i.e. Equation~\ref{eq:Hawkes-time-of-use}. Fix a segment $J=(\gamma,\gamma+w]$ whose effective interior length $w>0$ does not grow as the overall horizon $T\to\infty$. Let $\eta$ be a smooth parameter (e.g.\ a productivity level or ramp slope, distinct from the change time boundaries) entering only via $\kappa$ on $J$; outside $J$ the intensity does not depend on $\eta$. Assume for $\eta$ near $\eta_{\mathrm{true}}$:
\begin{enumerate}
\item[(i)] $\lambda_0>0$, so $\lambda(t;\eta)\ge \lambda_0$ for all $t\in J$.
\item[(ii)] There exists $C<\infty$ such that
\[
\sup_{t\in J}\E_{\eta_{\mathrm{true}}}\!\left[(\partial_\eta \lambda(t;\eta_{\mathrm{true}}))^2\right]\le C.
\]
See \cite{clinet2017statistical} for justification of this assumption. Note that change time parameters fail this type of regularity; their information is instead limited by the saturation bounds in Section~\ref{sec:mf-info}.
\end{enumerate}

    \begin{proposition}[Bounded information $\Rightarrow$ no unbiased consistency]
\label{prop:bounded-info}
Under \emph{(i)}-\emph{(ii)},
\[
\mathcal I_\eta(J)
=\E_{\eta_{\mathrm{true}}}\left[\int_{\gamma}^{\gamma+w}\frac{\{\partial_\eta \lambda(t;\eta_{\mathrm{true}})\}^2}{\lambda(t;\eta_{\mathrm{true}})}\,dt\right]
\ \le\ \frac{C}{\lambda_0}\,w
\ <\ \infty.
\]
Hence any unbiased estimator of $\eta$ has variance bounded below by a positive constant independent of $T$ and cannot be consistent as $T\to\infty$.
\end{proposition}

\begin{proof}
By (i), $\lambda(t;\eta_{\mathrm{true}})\ge \lambda_0$ for all $t\in J$, hence
\[
\mathcal I_\eta(J)
=\int_{\gamma}^{\gamma+w}\E_{\eta_{\mathrm{true}}}\!\left[\frac{\{\partial_\eta \lambda(t;\eta_{\mathrm{true}})\}^2}{\lambda(t;\eta_{\mathrm{true}})}\right]dt
\ \le\ \frac{1}{\lambda_0}\int_{\gamma}^{\gamma+w}\E_{\eta_{\mathrm{true}}}\!\left[\{\partial_\eta \lambda(t;\eta_{\mathrm{true}})\}^2\right]dt.
\]
By (ii), the integrand is uniformly bounded by $C$, so
\[
\mathcal I_\eta(J)\le \frac{1}{\lambda_0}\int_{\gamma}^{\gamma+w} C\,dt
=\frac{C}{\lambda_0}\,w
<\infty.
\]
The Cram\'er--Rao bound implies $\Var(\widehat\eta)\ge 1/\mathcal I_\eta(J)$ for any unbiased estimator. Since $\mathcal I_\eta(J)$ is bounded above by a constant independent of $T$, the lower bound $1/\mathcal I_\eta(J)$ is bounded away from $0$, so no unbiased estimator can be consistent as $T\to\infty$.
\end{proof}
    
    Practically, this means that estimation risk on short regimes saturates: asymptotic variance cannot be driven below a positive constant by increasing $T$. To improve accuracy one must either enlarge $w$ (often impossible) or inject information via structure (e.g., smoothness/shape constraints on $\kappa$, informative priors, or pooling across regimes). change point detection can still be reliable when the effect size is large, but precise estimation of fine-grained features (e.g., a short ramp slope or change time) is fundamentally limited. Our specification (shared background and kernel parameters across all regimes) further reduces incidental-parameter noise: $\lambda_0$ and $\beta$ are estimated from the full record, leaving the segment-specific curvature to concentrate on the $\kappa$-parameters that truly change.

    \begin{remark}[Variable infectivity analogue]
    Under the variable infectivity specification \eqref{eq:Hawkes-fixed-parent}, a parameter supported on $J=(\gamma,\gamma+w]$ affects $\lambda(t)$ for $t>\gamma+w$ through the kernel tail, but its contribution is kernel-weighted and therefore decays and is integrable. A proof follows by bounding $|\partial_\eta\lambda(t)| \le \sup_{u\in J}|\partial_\eta\kappa(u)|\int_J g(t-u)\,dN(u)$ and using $\int_0^\infty g(s)^2\,ds<\infty$ (in particular for exponential $g$) to show the expected information remains $O(w)$.
    \end{remark}
    
    \section{Stability and forgetting}
    \label{sec:stability}
    
    For change point problems we rely on two stability properties: past influence decays at a rate controlled by the relaxation time, justifying the boundary-layer view of information; and solutions remain well-behaved even with intermittent supercritical bursts, provided there is sufficient sub-criticality on average. The results below formalize these points for exponential kernels and support the prior design and calibration choices used in estimation.
    
    Assume an exponential kernel. Suppose there exist $L>0$ and $\delta\in(0,1)$ such that
    \begin{equation}
    \label{eq:avg-gap}
    \int_t^{t+L}\big(1-\kappa(u)\big) du\ \ge\ \delta L\qquad\text{for all }t\ge 0.
    \end{equation}
    Solving \eqref{eq:M-ode} by an integrating factor gives
    \[
    M(t)=e^{-\beta t+\beta\int_0^t \kappa(u) du}\left(
    M(0)+\beta\lambda_0\int_0^t e^{\beta s-\beta\int_0^s \kappa(u) du} ds\right).
    \]
    Set
    \[
    \phi(t):=e^{-\beta\int_0^t (1-\kappa(u)) du}.
    \]
    The average-gap condition \eqref{eq:avg-gap} implies
    \[
    \phi(t+L)=\phi(t) e^{-\beta\int_t^{t+L}(1-\kappa(u)) du}\ \le\ e^{-\beta\delta L} \phi(t)\qquad\forall t\ge 0.
    \]
    Iterating this inequality and writing $n=\lfloor t/L\rfloor$ gives
    \[
    \phi(t)=\phi(nL) e^{-\beta\int_{nL}^t(1-\kappa(u)) du}
    \ \le\ \Big(\sup_{r\in[0,L]}\phi(r)\Big) e^{-\beta\delta nL}
    \ \le\ C e^{-\beta\delta t},
    \quad
    C:=\Big(\sup_{r\in[0,L]}\phi(r)\Big) e^{\beta\delta L}.
    \]
    Therefore
    \[
    e^{-\beta t+\beta\int_0^t \kappa(u) du}=\phi(t)\ \le\ C e^{-\beta\delta t}.
    \]
    Substituting this bound into the integrating-factor representation of $M$ shows that the homogeneous part decays at rate $\beta\delta$, and hence $M(t)$ and $\bar\lambda(t)=\lambda_0+\kappa(t)M(t)$ forget initial conditions exponentially fast.
    
    Alternatively, stability can be viewed through the lens of a finite population with saturation as proposed in \cite{rizoiu2018sir}. In particular, let
    \[
    \lambda(t,m\mid\mathcal H_t)=\lambda_0 
    +\kappa(t)\left(1-\frac{N_t}{N_{\mathrm{pop}}}\right)\sum_{t_i<t} g(t-t_i),
    \]
    with $N_{\mathrm{pop}}\in(0,\infty)$. In this case, the mean quantities (for an exponential kernel) satisfy
    \[
    \bar N'(t)=\bar\lambda(t),\quad
    M'(t)=\beta\bar\lambda(t)-\beta M(t),\quad
    \bar\lambda(t)=\lambda_0+\kappa(t)\Bigl(1-\frac{\bar N(t)}{N_{\mathrm{pop}}}\Bigr)M(t).
    \]
    The closed system for $(\bar N,M)$ above uses the standard mean-field factorization 
    $\mathbb E[(1-N_t/N_{\mathrm{pop}})Z(t)]\approx (1-\bar N(t)/N_{\mathrm{pop}})M(t)$ and serves only as a calibration device. 
    For bounded $\kappa(t)$ and an exponential kernel, the right-hand side is smooth with at most linear growth, so the ordinary differential equation yields a unique global solution. Practically, we note that estimating $N_{\mathrm{pop}}$ is often non-trivial in applications of interest.

    \subsection{Explosive regimes and boundary amplification}
    \label{subsec:explosive}
    
    Supercritical (equivalently, explosive) regimes with \(\kappa(t)\ge 1\) are physically plausible in contagious-disease settings: transient periods with \(R_t>1\) can occur due to behavioural shocks, holidays, or delayed interventions \cite{you2020estimation}. Mathematically, such regimes are also compatible with finite-horizon modelling, and even with global stability provided sufficiently subcritical behaviour holds on average (e.g. the average-gap condition \eqref{eq:avg-gap} across blocks). However, for change point detection these regimes behave qualitatively differently from the subcritical case \(\kappa<1\): the boundary layer no longer ``forgets'' the past, but instead amplifies it, creating a fundamental identifiability obstruction when the change point time is unknown.
    
    To see this, assume an exponential kernel and consider a single supercritical segment on \((t^\ast,t^\ast+w]\) with
    \(\kappa(t)\equiv\kappa_+>1\) for \(t\in(t^\ast,t^\ast+w]\).
    Writing \(u:=t-t^\ast\ge 0\), the mean-field ODE \eqref{eq:M-ode} becomes
    \[
    M'(u)=\beta\lambda_0+\beta(\kappa_+-1)M(u),\qquad \bar\lambda(u)=\lambda_0+\kappa_+ M(u),
    \]
    whose solution is the supercritical analogue of \eqref{eq:M-const-solution}:
    \begin{equation}
    \label{eq:M-supercritical}
    M(u)=\Big(M(0^+)+\frac{\lambda_0}{\kappa_+-1}\Big)e^{\rho u}-\frac{\lambda_0}{\kappa_+-1},
    \qquad \rho:=\beta(\kappa_+-1)>0.
    \end{equation}
    Two features matter for inference. Of primary importance, in the subcritical case \(\kappa<1\), the first term in \eqref{eq:odesol} decays like
    \(\exp\{-\beta\int(1-\kappa)\}\), yielding exponential forgetting and the boundary-layer picture. For \(\kappa_+\ge 1\), the same integrating factor instead scales like \(\exp\{+\beta\int(\kappa_+-1)\}\), i.e.\ the influence of the initial condition \(M(0^+)=M(t^{\ast+})\) is multiplied by \(e^{\rho u}\) in \eqref{eq:M-supercritical}. Equivalently, excursions inherited
    from the pre-change history are not damped but rather exponentially scaled. Thus, in an explosive segment, the ``boundary term'' is not local: it propagates forward and dominates the entire segment.
    
    This leads to non identifiability. At first glance, an exponential regime might seem easy to fit because it produces many events later in time. The difficulty is that those later events primarily identify the product
    \(\big(M(0^+)+\lambda_0/(\kappa_+-1)\big)e^{-\rho t^\ast}\) appearing in \eqref{eq:M-supercritical}, not \(t^\ast\)
    itself. Indeed, for any \(\delta>0\),
    \[
    (t^\ast,\;M(0^+))\ \mapsto\ (t^\ast+\delta,\; \tilde M(0^+):=e^{\rho\delta}M(0^+))
    \]
    leaves the leading exponential factor \(\tilde M(0^+)e^{\rho(t-(t^\ast+\delta))}=M(0^+)e^{\rho(t-t^\ast)}\) essentially unchanged (up to the additive \(\lambda_0/(\kappa_+-1)\) term). This in turn creates a ridge geometry: shifting the onset later can be offset by inflating the inherited boundary state. Compounding this, early exponential growth is slow: expanding \eqref{eq:M-supercritical} at small \(u\) gives \(M(u)=M(0^+)+\beta\lambda_0\,u+O(u^2)\), so near the onset the path is almost linear and event counts are sparse. Consequently, the data carry little information about where the regime starts, precisely where
    \(t^\ast\) would need to be identified.
    
    It follows that likelihood-based estimation of a supercritical segment can be workable when the onset time \(t^\ast\) is fixed
    (or externally known), because \(\rho=\beta(\kappa_+-1)\) can then be learned from the macroscopic growth
    once \(u\) is not too small. In contrast, when \(t^\ast\) is unknown, boundary amplification destroys the
    localization mechanism discussed in Section~\ref{sec:info}: the influence of pre-change history does not
    decay into a short boundary layer, and the likelihood develops near-invariances that make \(t^\ast\) weakly
    identified even with large post-change counts. Understanding the resulting tradeoffs (the 
    ridges coupling \(\beta\), \(\lambda_0\), and \(\kappa_+\), as well as their interaction with the unknown boundary state \(M(0^+)\)) is an important direction for future work.
    
    Therefore we focus on subcritical segments \(\kappa<1\) for change point
    inference on real data with unknown change times, where forgetting holds and the boundary-correct likelihood
    remains well-conditioned.

  \section{Estimation, regularization, and change point detection}
\label{sec:estimation}

Section ~\ref{sec:mean-dynamics} demonstrates that post-change level parameters (e.g.\ $\kappa_2$) acquire curvature essentially linearly once beyond a short boundary layer, whereas the change time parameter $t^\ast$ is intrinsically local and exhibits rapid saturation of information. In practice this creates a non-standard optimisation geometry: the profile likelihood in $t^\ast$ is not only weakly curved, but also not smooth on a single realisation because moving $t^\ast$ across an event time reassigns that event to a different regime. Consequently, mode-based estimators for $t^\ast$ such as the maximum likelihood estimate (MLE) or maximum a posteriori (MAP) estimate tend to lock onto microscopic spikes rather than the macroscopic envelope described by mean-field information. This motivates an estimation strategy that (i) uses stable, mode-based updates for smooth parameters, and (ii) uses an averaging (smoothing) operator for change points. 

Before going further, we describe MLE estimation in our setting generally. For segment-wise evaluation we condition on the full pre-segment history through the carry-over term. Fix an observation window $[0,T]$ and event times $\{t_i\}_{i=1}^n$. For any candidate partition
$0=\gamma_0<\gamma_1<\cdots<\gamma_m<\gamma_{m+1}=T$,
write the segment index of time $t$ as $j(t)$ such that $t\in(\gamma_{j-1},\gamma_j]$.
For exponential kernels $g(s)=\beta e^{-\beta s}\mathbbm{1}_{\{s>0\}}$ it is convenient to carry forward the
one-dimensional boundary state
\[
S_j(\beta)\ :=\ \sum_{t_i<\gamma_{j-1}} \beta e^{-\beta(\gamma_{j-1}-t_i)},
\qquad
B_j(t)\ =\ S_j(\beta)\,e^{-\beta(t-\gamma_{j-1})}, \quad t>\gamma_{j-1},
\]
so that the filtered sum decomposes as $Z(t)=B_j(t)+W_j(t)$, where $W_j$ depends only on within-segment events.
Conditioning each segment on $S_j(\beta)$ yields the segment-conditional log-likelihood
\begin{equation}
\label{eq:condlik_new}
\ell_j(\theta;\gamma_{j-1:j})
=
\sum_{t_i\in(\gamma_{j-1},\gamma_j]}
\log \lambda(t_i;\theta)
-
\int_{\gamma_{j-1}}^{\gamma_j}\lambda(t;\theta)\,dt,
\end{equation}
with $\ell(\theta;\gamma_{1:m})=\sum_{j=1}^{m+1}\ell_j(\theta;\gamma_{j-1:j})$.
Operationally, this ``conditions on the full history'' while requiring only the scalar carry-over state $S_j(\beta)$.  The carry-over likelihood centres each segment 's score and yields a boundary-correct, modular likelihood.

In practice, parametric estimation of the conditional intensity parameters of Hawkes processes via MLE is generally computationally expensive \cite{kresin2023parametric}, and though such estimates are consistent \cite{ogata1978estimators}, they can be unstable in finite observation windows \cite{reinhart2018review,veen2008estimation}. Further, in the case of our model specification, the mean field analysis of Section \ref{sec:mf-info} demonstrates that the curvature in a neighbourhood of the true parameter vector  necessary for MLE convergence is absent.   We therefore suggest fitting by alternating (A) a MAP update of the continuous parameters at a  fixed change point configuration,
and (B) a center-of-mass (CoM) update of the change point(s) under the current parameters.
This is a coordinate-descent scheme in which the usual hard assignment of latent change points is replaced by a
smoothed summary (posterior mean on a grid), precisely to neutralise the jagged likelihood geometry in the $t^\ast$ profile. This approach is further justified in that there is often an empirical gap between MAP and CoM for $t^\ast$ which itself is a useful diagnostic of limited identifiability and of the discretisation gap between the continuous-time surrogate curvature $\mathcal I_{t^\ast}^{\mathrm{mf}}$ and the curvature recoverable from a single realisation.

We now describe the implementation of this scheme. Let $\theta$ denote the continuous parameters (e.g.\ $\lambda_0,\beta$ and segment productivities), and let $t^\ast = (\gamma_1, \dots, \gamma_{m})$ denote the vector of change point locations to be estimated.
Fix priors $\pi(\theta)$ and $\pi(t^\ast)$ (we generically recommend weakly informative log-normal priors for $\lambda_0,\beta$
and Beta priors for $\kappa$-levels). To implement our estimation method, we first find the MAP for the smooth parameters. In particular, given the current change point estimate $t^{\ast{(r)}}$, update
\begin{equation}
\label{eq:map_step}
\theta^{(r+1)} \in \arg\max_{\theta}\ \Big\{\ell(\theta;t^{\ast{(r)}}) + \log\pi(\theta)\Big\}.
\end{equation}
This step is well supported by curvature because $\theta$ enters smoothly, and (crucially) $\lambda_0$ and $\beta$
are shared across regimes, pooling information globally. We then use CoM for the change points. Given $\theta^{(r+1)}$, evaluate the (log-)posterior over a fine grid $\mathcal G$ of candidate change points:
\[
\log \widetilde \pi(t^\ast\mid\text{data},\theta^{(r+1)})\ \propto\ \ell(\theta^{(r+1)};t^\ast) + \log\pi(t^\ast),
\qquad t^\ast\in\mathcal G.
\]
Define weights $w(t^\ast)\propto \exp\{\ell(\theta^{(r+1)};t^\ast)+\log\pi(t^\ast)\}$ and update the change point by its
posterior centre of mass:
\begin{equation}
\label{eq:com_step}
t^{\ast{(r+1)}} \ :=\ \sum_{t^\ast\in\mathcal G} t^\ast\, w(t^\ast)
\quad\text{(component-wise if $t^\ast$ is vector-valued)}.
\end{equation}
Because $\ell(\cdot;t^\ast)$ is piecewise-smooth with kinks and jumps as $t^\ast$ crosses event times,
\eqref{eq:com_step} is essential: it averages over the local spikes and returns the stable centre of the macroscopic
envelope predicted by the mean-field information calculations in Section~\ref{sec:mf-info}, supplying the missing regularity that mode-based updates lack.

\subsection{Model selection}

The above estimator suggests the following approximation to model evidence as long as there is enough data in each segment to warrant a BIC approximation to model evidence on the segment. In particular, for each model $M$,
\begin{align}\label{eq:int_evidence} p(\mathrm{data}|M)=&\int \int p(\mathrm{data}|\theta,t^\ast,M)\pi(\theta|M,t^\ast)\,d\theta \,\pi(t^\ast|M)\,dt^\ast\,\pi(M)\\
\approx &c_{\mathrm{max}}\int \frac1{c_\mathrm{max}}\exp\Big\{
\underbrace{\ell(\widehat\theta;t^\ast)+\log\pi(\widehat\theta|t^\ast,M)}_{\text{profile log-posterior}}
\;-\;\underbrace{\tfrac{p_M}{2}\log n}_{\text{penalty}}
\Big\}\,
\pi(t^\ast\mid M)\,dt^\ast\,\pi(M),
\end{align}
where we are using Laplace's method to approximate the inside integral, $c_{\mathrm {max}}$ is of the order of the exponential to make numerical computation of the integral feasible, $t^\ast$ is a vector of change points, $p_M$ is the number of other parameters of the model $M$, and $\widehat \theta$ maximises the profile log-posterior given $t^\ast$; i.e., given by MAP on each segment as above. 


In practice we approximate the integral in \eqref{eq:int_evidence} by Monte Carlo samples of $t^\ast$ from a simple prior
(e.g.\ uniform with minimum-gap constraints), use a short inner optimisation for $\widehat\theta(t^\ast)$, and compute the
log-average stably via log-sum-exp. We place a mild structural prior on the number of change points, e.g.\ $\pi(M)\propto \mathrm{Poisson}(K;\lambda_{\mathrm{segs}})$, where $K$ is the number of segments of $M$. 

\subsection{Simulation study}
\label{sec:sim}

All simulations use exponential kernels $g(t)=\beta e^{-\beta t}\mathbbm{1}_{\{t>0\}}$.
We focus on variable susceptibility specification \eqref{eq:Hawkes-time-of-use} because it most clearly exhibits the
likelihood discontinuities that motivate CoM smoothing (and thereby the above described estimation strategy). We additionally contrast with the variable infectivity specification \eqref{eq:Hawkes-fixed-parent} in the likelihood profile experiment below. We first demonstrate the underlying pathology that makes change point detection in the variable productivity setting difficult. This is demonstrated most clearly by visualising the profile log-likelihood as a function of the change time $t^\ast$. This profile can be jagged and discontinuous on a single realisation. The discontinuities occur because shifting $t^\ast$ across an event time switches which productivity applies at that event, creating jumps in the log-intensity sum and in the compensator.

\begin{figure}[!ht]
  \centering
  \includegraphics[width=0.98\linewidth]{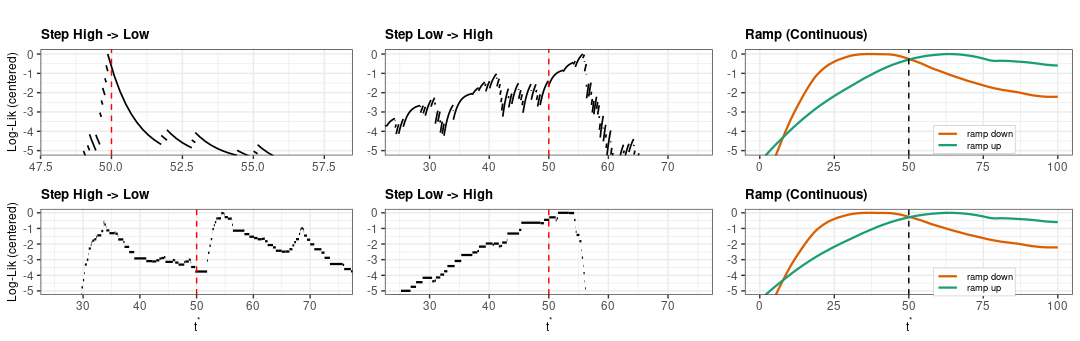}
  \caption{Change-time profile log-likelihoods (centred by subtracting the maximum) under three scenarios.
  Top row: variable susceptibility. Bottom row: variable infectivity.  Dashed vertical lines mark the true $t^\ast=50$.}
  \label{fig:sim_profiles}
\end{figure}

We fix $T=100$, $\lambda_0=1$, $\beta=1$, and a true change time $t^\ast=50$. We consider three productivity scenarios: high $\to$ low wherein $(\kappa_1,\kappa_2)=(0.75,0.25)$, low $\to$ high wherein $(\kappa_1,\kappa_2)=(0.25,0.75)$, and a continuous ramp wherein $\kappa(t)=\kappa_0$ for $t<t^\ast$ and $\kappa(t)=\kappa_0+s\,(t-t^\ast)$ for $t\ge t^\ast$,
with $(\kappa_0,s)=(0.25,0.005)$ (ramp up) and $(0.75,-0.005)$ (ramp down). Under model specification \ref{eq:Hawkes-time-of-use}, a step in $\kappa$ produces an instantaneous intensity jump and the change-time profile becomes a collection of
disconnected segments separated by event-induced jumps. Under the variable infectivity model, productivity is tied to parent times;
step changes therefore enter through the kernel-weighted offspring contributions and the resulting $t^\ast$ profile is
substantially smoother. For ramps, $\kappa(t)$ is continuous and both specifications yield smooth profiles, but still demonstrate flatness near the true value as well as persistent skewness.

We next demonstrate that simultaneous recovery of continuous parameters and the change point(s) is possible using the above proposed alternating MAP/CoM when the model is correctly specified. We simulate 500 single-change variable susceptibility Hawkes processes in the low $\to$ high regime, using $\log\lambda_0\sim \mathcal N(0,2^2)$ and $\log\beta\sim \mathcal N(0,2^2)$ (log-normal on $\lambda_0,\beta$), and $\kappa_j\sim \mathrm{Beta}(2,2)$ independently. We report MAP estimates for $(\kappa_1,\kappa_2,\lambda_0,\beta)$ and the CoM estimate for $t^\ast$ in Figure \ref{fig:sim_param_recovery}.

Finally, we demonstrate that the criterion in Equation \ref{eq:int_evidence} is effective for model selection when change point likelihoods are jagged and curvature is minimal or non-existent. We simulate  200 realisations allowing for $0$, $1$, or $2$ change points on $T=200$, where  $\lambda_0\sim \mathrm{Unif}(1,5)$ and $\beta\sim \mathrm{Unif}(1,5)$. Each candidate model ($K=0,1,2$) is fit by alternating MAP/CoM then scored. The results are summarised in Figure \ref{fig:sim_integrated_selection}.

\subsection{Information growth and likelihood shape}
\label{sec:mf-sim}

We simulated single-change Hawkes paths under the variable susceptibility specification and compared the mean-field formulas of Section~\ref{sec:mf-info}. In all experiments the kernel is exponential with $\beta=1$ and the change time is $t^\ast=50$. We consider two step configurations, an up-step $(\kappa_1,\kappa_2)=(0.25,0.75)$ and a down-step $(0.75,0.25)$, and take the post-change window length $\Delta_{\text{post}}$ in units of the post-change relaxation time $\tau_2=1/[\beta(1-\kappa_2)]$.

For the information plots in Figure~\ref{fig:mf-info-step} we vary the baseline rate over $\lambda_0\in\{1,2,4,8\}$ and simulate $R=2000$ independent paths for each $(\lambda_0,\kappa_1,\kappa_2)$ configuration. For the level parameter $\kappa_2$ we compute the empirical Fisher information on $[t^\ast,t^\ast+\Delta_{\text{post}}]$ as the variance of the score (black circles), and compare it with (i) the closed-form mean-field prediction $\mathcal I_{\kappa_2}^{\mathrm{mf}}(\Delta)$ in \eqref{eq:Ikappa2-mf-closed} (solid blue lines), and (ii) the Fokker-Planck refinement $\mathcal I_{\kappa_2}^{\mathrm{mf+}}(\Delta)$ in \eqref{eq:Ikappa-mf-plus} (red dashed lines). For the change time $t^\ast$ we use the precision $1/\Var(\hat t^\ast)$ of the path-wise MLE as a proxy for curvature (black circles) and compare it with the mean-field expression in \eqref{eq:Itstar-mf}: the magenta dashed curves show the smooth relaxation contribution only, while the purple solid curves add the singular jump information \eqref{eq:I-jump}, giving the bound \eqref{eq:Itstar-mf-bound}.

The level-information panels confirm the linear-plus-transient structure \eqref{eq:Ikappa2-mf-asymp}. After a short boundary layer of length $O(\tau_2)$, the empirical information grows essentially linearly with $\Delta_{\text{post}}$, with slope close to $\Lambda_2=\lambda_0/(1-\kappa_2)$ and nearly independent of the step direction. The mean-field curve already tracks this behaviour well, and the Fokker-Planck correction is only visible for small windows and lower baselines, where carry-over variability inflates $\Var(Z_t)$ immediately after the change. This is exactly the Jensen gap predicted after \eqref{eq:Ikappa2-mf-closed}: because $z\mapsto z^2/(\lambda_0+\kappa_2 z)$ is convex, replacing the random $Z_t$ by its mean $M_t$ underestimates information, with the largest discrepancy in the down-step case where the pre-change regime is more variable. Practically, the linear slopes in Figure~\ref{fig:mf-info-step} calibrate how many relaxation times of post-change data are needed to reach a target curvature for $\kappa_2$.

The change-time panels illustrate the locality of $t^\ast$. In all configurations the empirical precision quickly saturates in $\Delta_{\text{post}}$, and the saturation level grows roughly linearly in $\lambda_0$, as predicted by \eqref{eq:Itstar-mf}-\eqref{eq:Itstar-mf-bound}. Comparing up- and down-steps, the plateau is much higher for the down-step $(0.75\to0.25)$: it is easier to localise a drop from a busy regime into a quiet one than a jump from a quiet regime into a busy one, because the information for $t^\ast$ is dominated by a small boundary layer around the change rather than by the long post-change horizon.

While the mean-field surrogate $\mathcal I_{t^\ast}^{\mathrm{mf}}(\Delta_{\mathrm{post}})$ accurately predicts the shape of information growth and its saturation with post-change window length, Figure~\ref{fig:mf-info-step} shows a persistent gap between this theoretical curvature and the realised empirical precision $1/\Var(\hat t^\ast)$ (black circles). The mean-field quantity is the curvature of the log-likelihood one would obtain if the conditional intensity $\bar\lambda(t)$ were observed continuously; because a realisation is comprised of discrete arrival times, the path-wise log-likelihood is a jagged function of $t^\ast$ with kinks at every event time. This discreteness breaks the local asymptotic normality approximation that would otherwise equate curvature and Fisher information for change times, and effectively limits localisation accuracy to the scale of inter-arrival times. Consequently, we view $\mathcal I_{t^\ast}^{\mathrm{mf}}$ as describing the information content of the macroscopic envelope or posterior ridge, and the plateau of $1/\Var(\hat t^\ast)$ as the amount of that envelope curvature that is actually recoverable from a single discretised path.

To study the full likelihood geometry we fix $\lambda_0=1$ and simulate $R=500$ paths for each step direction. For each path we evaluate the log-likelihood $\ell(\kappa_2,t^\ast)$ on a grid and compute three estimators: the MLE (likelihood mode), a MAP estimator under a weak Gaussian prior centred at the true $(\kappa_2,t^\ast)$, and the posterior centre of mass (CoM). Figure~\ref{fig:relax-grid} displays the mean likelihood surface and one-dimensional profiles through the peak.

\begin{figure}[!h]
  \centering
  \includegraphics[width=\linewidth]{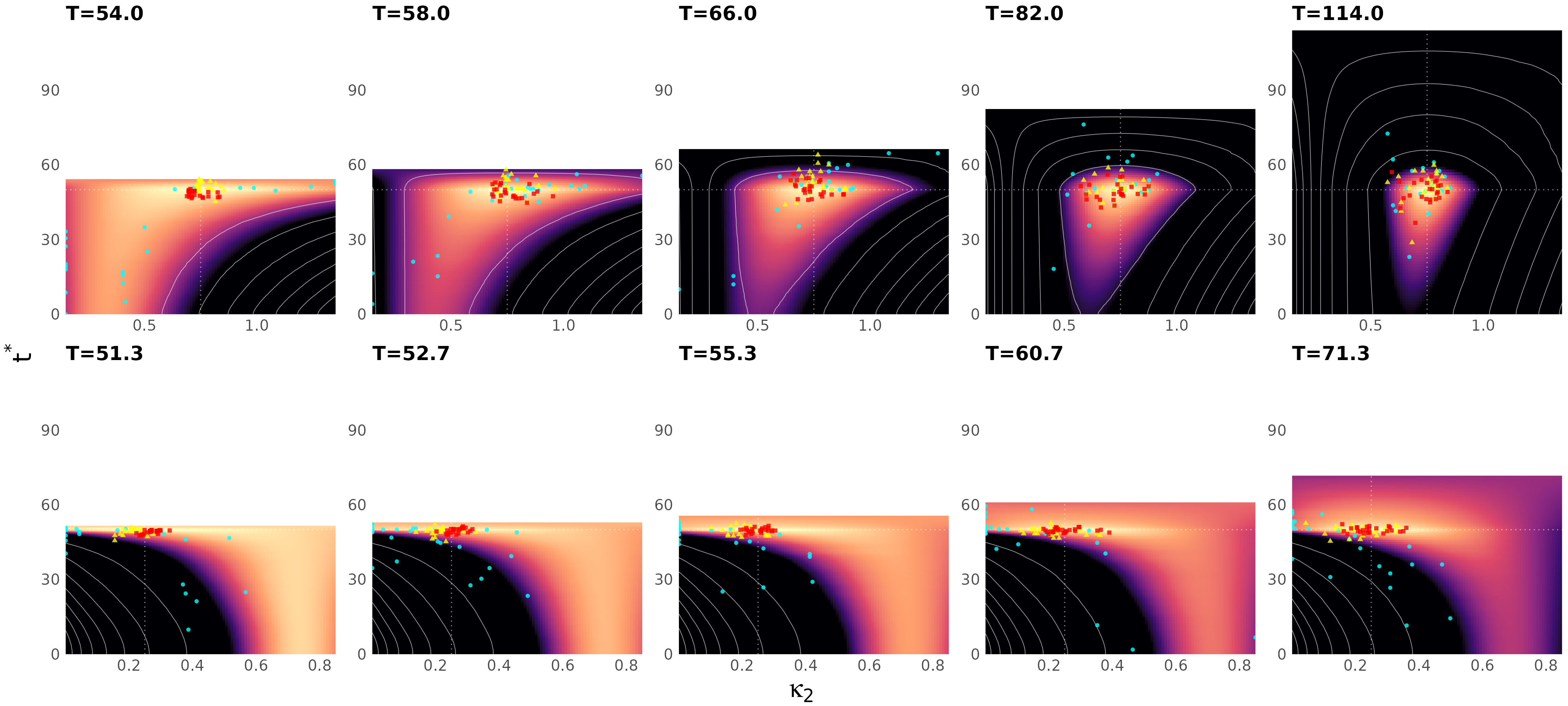}
  \caption{\textbf{Likelihood surface evolution over time.} Mean log-likelihood surfaces for post-change parameters $(\kappa_2, t^*)$ under two scenarios: Low-to-High ($\kappa: 0.25 \to 0.75$, top row) and High-to-Low ($\kappa: 0.75 \to 0.25$, bottom row). Columns correspond to increasing post-change observation windows defined by multiples of the post-change relaxation time $\tau = [\beta(1-\kappa_2)]^{-1}$, specifically $1\tau, 2\tau, 4\tau, 8\tau,$ and $16\tau$ (left to right). The surfaces illustrate the transition from the transient regime (left) to the asymptotic steady state (right). In the early transient phase, surfaces are diffuse and highly skewed, particularly in the High-to-Low scenario where the penalty for late detection is linear rather than logarithmic. As data accumulates beyond the boundary layer, the likelihood concentrates tightly into a symmetric ellipse around the true parameters (dashed crosshairs). Individual estimates for 25 realizations are overlaid: MLE (cyan circles), MAP (yellow triangles), and CoM (red squares). Note that in the transient regime, MLEs disperse widely along the ridge of the likelihood surface, while the CoM estimator remains stable and tightly clustered near the true change point.}
  \label{fig:relax-grid}
\end{figure}

In Figure~\ref{fig:relax-grid}, the likelihood surfaces are computed over the post-change point parameters $(\kappa_2,t^\ast)$ while holding the background rate, decay, and pre-change productivity fixed at their true values ($\lambda_0=1$, $\beta=1$, and $\kappa_1\in\{0.75,0.25\}$). This isolates the geometry of the transition itself. In the early transient phase (left columns), the extended right-hand tail of the surface in the down step case $(0.75\to0.25)$ reflects the fact that pushing the candidate change $t^\ast$ later extends the high-productivity regime into a genuinely quiet period. On $(t^\ast_{\mathrm{true}},t^\ast)$ the model over predicts the rate, but because the true intensity is low there are few events, so the log sum $\sum_{t_i}\log\lambda(t_i)$ changes little and the penalty is dominated by the compensator term $-\int\lambda(t)\,dt$. The mean likelihood surface illustrates the role of the compensator in the absence of data. In the down-step scenario ($0.75\to0.25$), the noticeable gradient along the $\kappa_2$ axis arises because the true post-change regime is quiet (lack of points). As $\kappa_2$ increases (moving right), the model predicts excess intensity without events to justify this higher rate. As a consequence, the compensator purely penalizes the objective, driving the likelihood down. By contrast, the steep drop in the bottom-left corner reflects the opposing penalty: assigning a low rate to the busy pre-change period.

As the post change window extends (moving left to right in Figure~\ref{fig:relax-grid}), the geometry transforms. For the level parameter, the surfaces tighten vertically, consistent with the linear growth of information $\mathcal I_{\kappa_2}^{\mathrm{mf}}$. The change-time geometry, in contrast, remains strongly skewed in the short-window limit. This skewness reflects the jagged path-wise likelihood: as $t^\ast$ moves, the likelihood has kinks whenever the proposed boundary crosses an event time. This piecewise-smooth geometry explains why the observed Hessian with respect to $t^\ast$ on a single path is so unreliable: the local second derivative is dominated by a handful of events and does not represent the saturated information level in \eqref{eq:Itstar-mf-bound}.

The direction of skewness is clearly visible in the transient columns of Figure~\ref{fig:relax-grid}: in the down-step $(0.75\to0.25)$ the surface has a long right tail. In that case, it is comparatively cheap, in log-likelihood units, to move the change time later than the truth, but expensive to move it earlier. Conversely, in the up-step $(0.25\to0.75)$ the pattern reverses and the tail points to the left, making it cheap to declare the change too early and costly to declare it too late. These asymmetries are precisely what drive the systematic early/late bias seen in the simulations.

The evolving geometry in Figure~\ref{fig:relax-grid} means that the three estimators respond very differently to this landscape. The MLE only finds the mode based on the path-wise likelihood and therefore chases local peaks along the ridge of the surface. When curvature is weak or highly anisotropic (left columns), the MLEs (cyan dots) disperse widely along the ridge and can cluster on its edges. The MAP (yellow triangles) is still a mode, but the weak Gaussian prior adds a small, isotropic curvature term, pulling estimates toward the true $(\kappa_2,t^\ast)$. The CoM (red squares) is the posterior centre of mass. It averages over the entire posterior rather than just locating its maximum. In these simulations the CoM is visually the most stable: it is less sensitive to the jagged local geometry and to the long one-sided tails than either mode-based estimator, and concentrates tightly near the true parameters. CoM and MAP coincide when the posterior becomes approximately symmetric and unimodal (as in the asymptotic right-most columns), but diverge when the posterior is skewed, precisely the regime where $t^\ast$ is hardest to identify.

 \section{Application to contagious disease spread in The Gambia}\label{sec:app}

We apply our model to analyse the incidence of invasive pneumococcal disease (IPD) using data from the Basse Health and Demographic Surveillance System (BHDSS) in the Upper River Region of The Gambia \cite{mackenzie2012monitoring}. The surveillance population included all BHDSS residents aged 2 months or older; our analysis focuses on cases among children aged two months to 14 years. In 2017, the total population was 181,740, of whom 85,279 were aged 2 months to 14 years \cite{mackenzie2021impact}. We analyse IPD cases recorded between May 2007 and November 2015, spanning the introduction of the seven-valent pneumococcal conjugate vaccine PCV7 in August 2009 and PCV13 in May 2011. PVC13  was administered in three primary doses at two, three, and four months of age. Over the 8-year study period, 348 IPD cases were confirmed among children aged 2 months to 14 years. Patients were screened by nurses and referred to clinicians following standardized criteria, and IPD cases were defined as those from whom \textit{Streptococcus pneumoniae} was isolated from a normally sterile site \cite{mackenzie2012monitoring}. Previous analyses \cite{mackenzieEffectIntroductionPneumococcal2016, mackenzie2021impact} found an 80\% reduction in IPD incidence among children under five between the pre-vaccine period and ten years post-introduction. We apply our method to detect change points in the incidence series and assess whether these data-driven estimates align with the vaccine rollout dates.
As discussed in Remark~\ref{rem:fixed-parent-mf}, the mean-field and information-geometric results in Sections~\ref{sec:mean-dynamics}--\ref{sec:mf-info} were derived under the variable susceptibility specification \eqref{eq:Hawkes-time-of-use}. We adopt this specification for the analysis, as placing the productivity term outside the kernel sum aligns with the mechanism of a vaccine campaign: it represents a population-wide reduction in susceptibility that scales the transmission potential of all currently active cases. We also fitted the variable infectivity specification \eqref{eq:Hawkes-fixed-parent} and observed no discernible difference in the estimated change points or productivity levels. This robustness is expected because the estimated kernel decay time is significantly shorter than the duration of the identified regimes; consequently, the transient differences between the instantaneous jump of the variable susceptibility model and the smooth relaxation of the variable infectivity model are negligible relative to the timescale of the analysis.

Temporal Hawkes models were fit to discrete daily case data by uniformly assigning continuous times within each day. To capture complex dynamics such as gradual adoption or waning efficacy, we explored combinations of constant (C), linear ramp (R), and exponential (E) segments in the piecewise productivity function. Continuous parameters (background rate, temporal decay, and productivity parameters) were fit using the boundary-correct segment-conditional likelihood \eqref{eq:condlik_new} and the alternating MAP/CoM scheme described in Section~\ref{sec:estimation}. Concretely, for a given candidate change point configuration we compute a MAP update of $(\lambda_0,\beta)$ and the segment productivity parameters under weakly informative priors, and then update the change point(s) using a centre-of-mass (CoM) step computed from the change point posterior evaluated on a fine grid. This averaging step is essential in these sparse data: the path-wise change-time profile is jagged (kinks at event times), so mode-based change point updates can lock onto microscopic spikes rather than the macroscopic envelope predicted by the mean-field information calculations. Model comparison was performed using the integrated evidence criterion of \eqref{eq:int_evidence}, which marginalises over change points to stabilise selection.

The data are sparse, with 348 cases reported over the seven-year period. To ensure stability, particularly during periods of low transmission, we constrained $\kappa \in (0,1)$ and used priors that penalise values near the boundaries to prevent estimates from collapsing to zero. Figure~\ref{fig:model_ranking} displays the top 25 models ranked by integrated evidence. A consistent pattern emerges across the highest-ranked specifications: the productivity remains stable in the early period and undergoes a significant decline beginning around $t=1000$. The best-performing model with change points, denoted \texttt{2\_CP\_CRC}, specifies a Constant-Ramp-Constant productivity profile. This result suggests that a model allowing for a gradual decline (ramp) is statistically preferred over instantaneous step-changes for these data.

\begin{figure}[!htbp]
    \centering
    \includegraphics[width=0.98\linewidth]{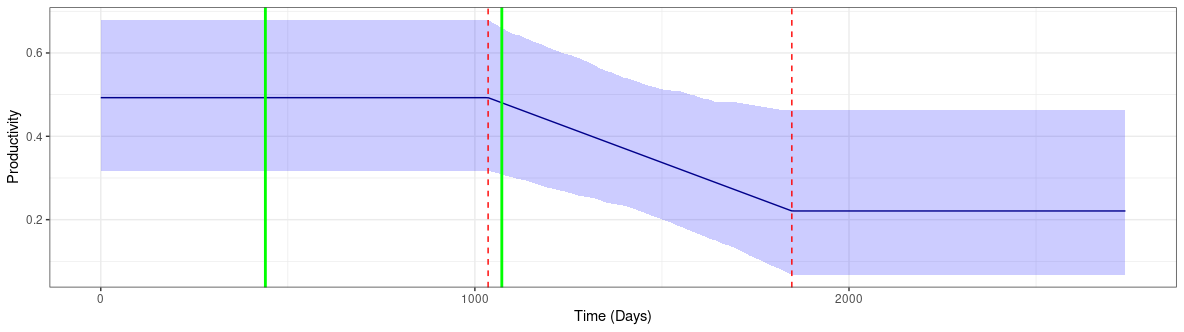}
    \caption{Posterior estimate for the best-fitting model (\texttt{2\_CP\_CRC}). The blue shaded region indicates the 95\% credible interval for $\kappa(t)$. Vertical green lines mark the vaccine rollouts (PCV7 left, PCV13 right). Red dashed lines indicate the estimated change points. The model identifies a stable pre-intervention productivity followed by a linear decline starting almost concurrently with the PCV13 introduction.}
    \label{fig:best_model}
\end{figure}

The fit of the best model is shown in Figure~\ref{fig:best_model}. The productivity $\kappa(t)$ is estimated to be approximately constant at $\kappa \approx 0.48$ during the pre-intervention and PCV7 periods. A structural break is identified at the onset of a linear decline (ramp) segment. Notably, this first change point aligns closely with the introduction of PCV13, suggesting that the second vaccine drive was the primary driver of the reduction in transmission potential. The ramp segment indicates a gradual reduction in contagion over approximately 800 days, stabilizing at a new lower level of $\kappa \approx 0.21$. This gradual decline is consistent with the time required to achieve high vaccine coverage and herd immunity effects.

The uncertainty of the detected regime shifts is characterized using the joint High Posterior Density (HPD) region shown in Figure~\ref{fig:joint_hpd}. The density contours reveal a distinct asymmetry in precision between the two events. The first change point (CP1) is tightly constrained along the horizontal axis, centered at $t \approx 1050$ and situated entirely prior to the PCV13 introduction date (dotted vertical line). This localization confirms that the onset of the decline is well-identified by the data. In contrast, the region is vertically elongated, indicating significantly higher uncertainty regarding the second change point (CP2), which spans a broad range from $t \approx 1750$ to $1950$. This shape implies that while the start of the intervention effect is sharp, the transition to the subsequent plateau is gradual and less precise in time.

\begin{figure}[!htbp]
    \centering
    \includegraphics[width=0.8\linewidth]{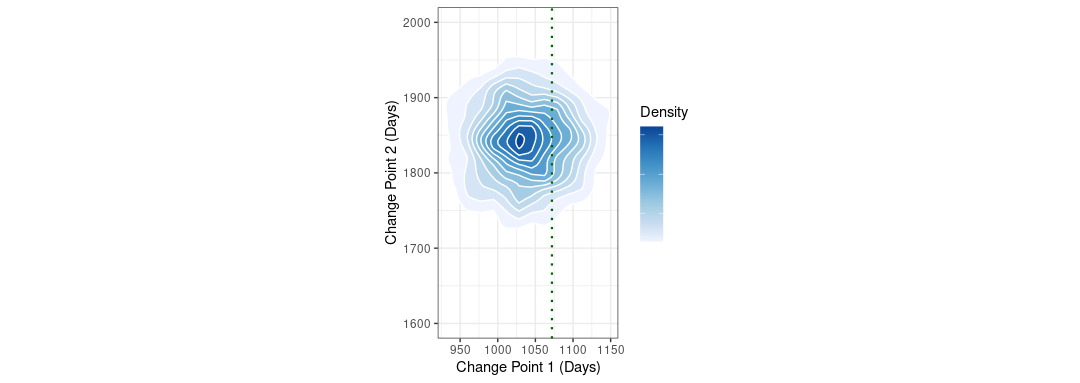}
    \caption{Joint HPD region for the first two change points. The contour plot visualizes the bivariate posterior probability, where darker shades represent regions of higher density. The vertical dotted green line marks the introduction of PCV13 ($t=1072$). The shape of the distribution highlights the contrast in temporal precision: the first change point (x-axis) is sharply localized immediately prior to the intervention, whereas the second change point (y-axis) exhibits significantly larger variance.}
    \label{fig:joint_hpd}
\end{figure} 

The application demonstrates the utility of the proposed framework: by allowing variable productivity and using a stable, evidence-based selection criterion, we recovered a biologically plausible trajectory of disease control marked by a gradual decline following the intervention.

\section{Conclusion}
    This paper develops Hawkes processes in which the endogenous amplification $\kappa(t)$ varies piecewise while the background rate and kernel parameters are shared across segments. A mean-field analysis yields explicit relaxation laws for the post-change mean, shows boundary-layer localization of information, and leads to a closed-form deterministic surrogate for Fisher information after a change: information for post-change levels grows essentially linearly once past a short transient, whereas information for the change time saturates. Sharing $\lambda_0$ and kernel parameters across all segments stabilizes curvature and mitigates confounding. This perspective yielded a novel estimation strategy and complementary model selection criteria. Simulations corroborate the surrogate information rates and information geometry, and the application to invasive pneumococcal disease in The Gambia produces interpretable change points that align with vaccine rollout, despite sparse counts and short effective regimes.
    
    Several extensions are natural. First, replacing the exponential kernel with heavy tailed or multi scale kernels (e.g., power laws) would broaden applicability and require generalized relaxation and information formulas; the mean-variance ODEs and the master equation viewpoint offer a route to such results. Second, jointly time-varying background and productivity could clarify identifiability trade-offs and allow covariate-driven structure in $\kappa(t)$, marks, or space-time ($\kappa(t,m,\mathbf{x})$). Third, online detection with stopping rules calibrated by the surrogate information could turn the framework into a real-time monitor, including handling supercritical bursts and finite population saturation. Finally, building on the work of \cite{kanazawa2020field} to develop multivariate extensions Fokker-Planck equations is of theoretical interest. 
    
    \section{Acknowledgments} We would like to acknowledge Angus Henderson for his work on developing a stochastic EM algorithm for change point detection, as well as Philip Hill and Grant Mackenzie for giving access and insight with respect to the Gambian disease data discussed in Section~\ref{sec:app}.
    \bibliographystyle{plain} 
    \bibliography{bib}
\newpage
    \appendix

    \section{Figures}

    \begin{figure}[!ht]
  \centering
  \includegraphics[width=0.9\linewidth]{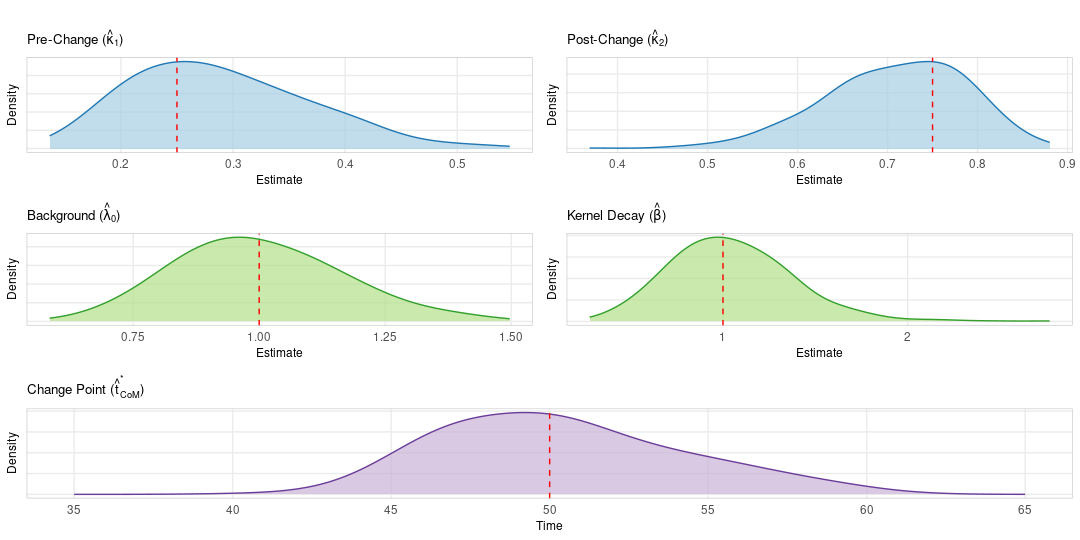}
  \caption{Parameter recovery across replicated variable susceptibility simulations using alternating MAP--CoM.
  Vertical dashed lines mark the true parameter values.}
  \label{fig:sim_param_recovery}
\end{figure}

\begin{figure}[!ht]
  \centering
  \includegraphics[width=0.3\linewidth]{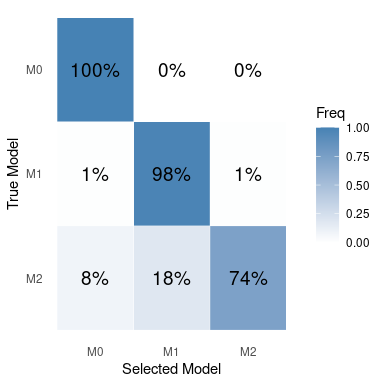}
  \caption{Model selection using integrated evidence (marginalised over change points) after alternating MAP/CoM fitting.
  Marginalisation stabilises selection by averaging over the jagged change point likelihood landscape rather than relying on
  a single maximiser.}
  \label{fig:sim_integrated_selection}
\end{figure}

\begin{figure}[!h]
  \centering
  \includegraphics[width=0.48\linewidth]{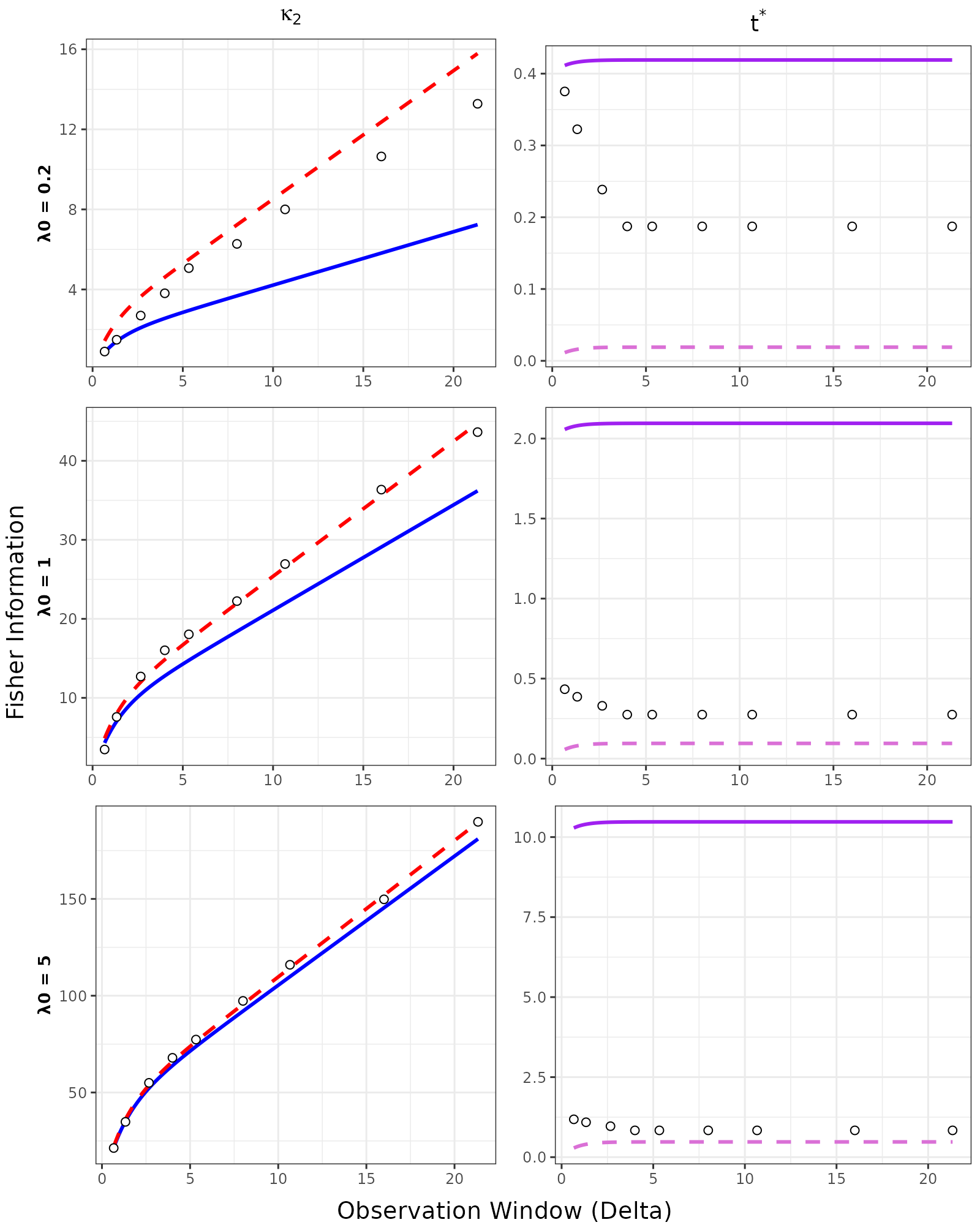}\hfill
  \includegraphics[width=0.48\linewidth]{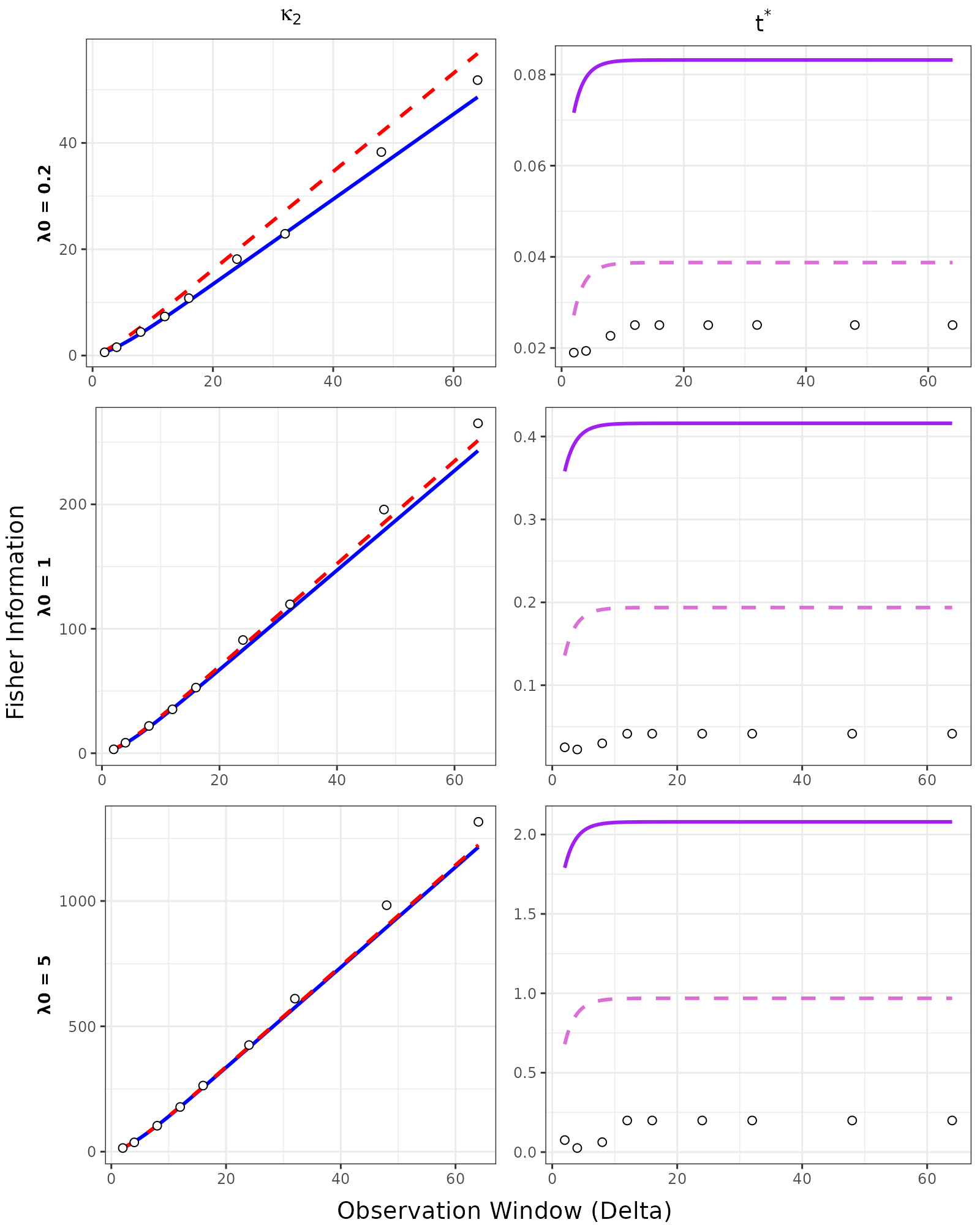}
  \caption{Information for level and change-time parameters as a function of post-change window length $\Delta_{\text{post}}$, in units of the post-change relaxation time $\tau_2$. Each panel corresponds to one baseline $\lambda_0\in\{0.2,1,5\}$ (rows). Left: down-step from $\kappa_1=0.75$ to $\kappa_2=0.25$. Right: up-step from $\kappa_1=0.25$ to $\kappa_2=0.75$. Within each panel, the left subplot shows information for the post-change level $\kappa_2$: solid blue lines are the mean-field surrogate $\mathcal I_{\kappa_2}^{\mathrm{mf}}$ from \eqref{eq:Ikappa2-mf-closed}, red dashed lines are the Fokker-Planck refinement \eqref{eq:Ikappa-mf-plus}, and black circles show Monte Carlo estimates from score variances. The right subplot shows information for the change time $t^\ast$: magenta dashed lines show the smooth relaxation term in \eqref{eq:Itstar-mf}, purple solid lines add the jump term \eqref{eq:I-jump} (yielding the bound \eqref{eq:Itstar-mf-bound}), and black circles are empirical precisions $1/\Var(\hat t^\ast)$. Level information grows essentially linearly with slope close to $\Lambda_2$, whereas change-time information saturates rapidly at a level that increases with $\lambda_0$ and depends strongly on the step direction.}
  \label{fig:mf-info-step}
\end{figure}

\begin{figure}[htbp]
    \centering
    \includegraphics[width=\linewidth]{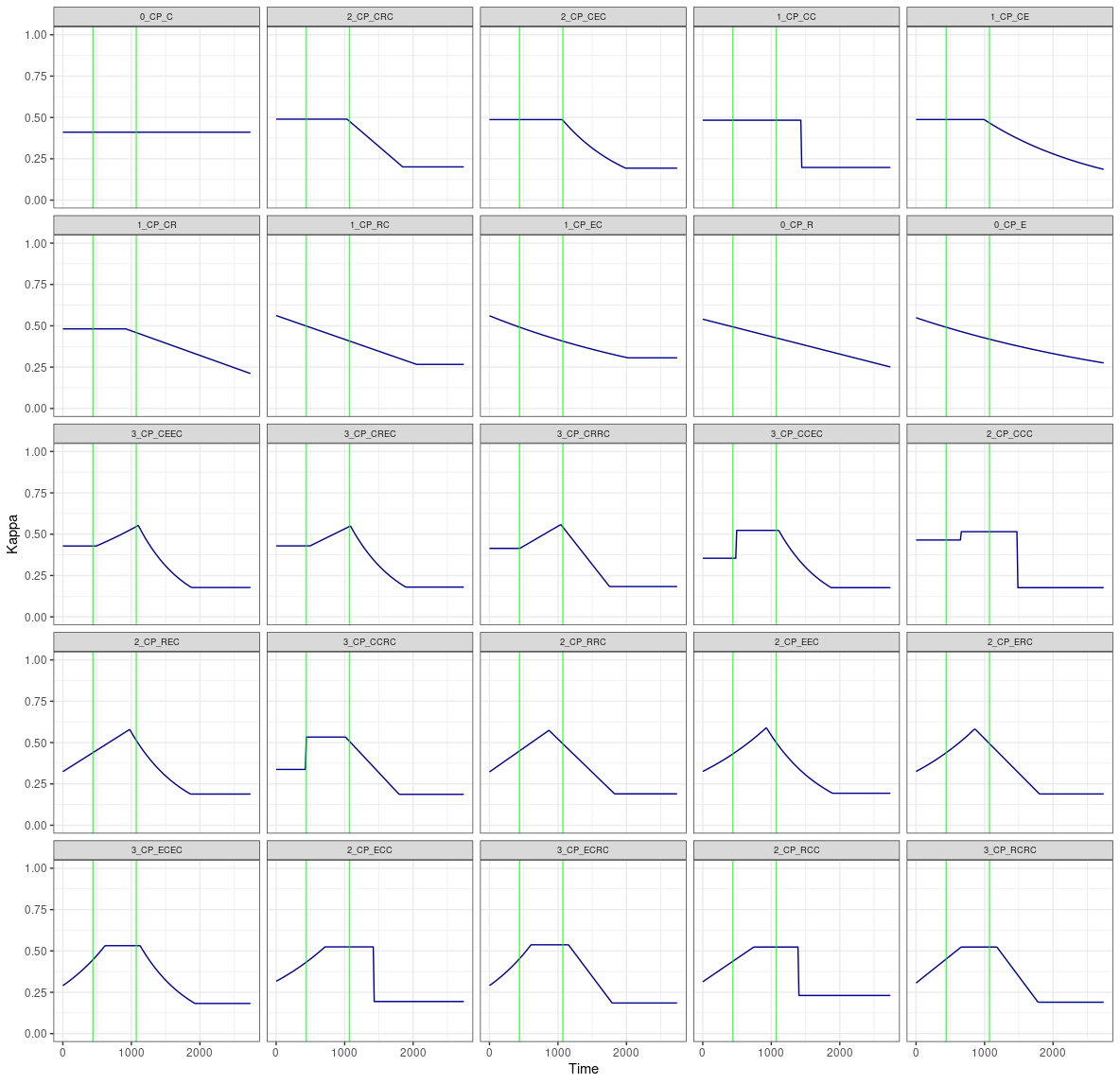}
    \caption{Top 25 models ranked by Integrated Evidence. The models are named by the number of change points and the sequence of segment types (C=Constant, R=Ramp, E=Exponential). Vertical green lines indicate the introduction of PCV7 (day 440) and PCV13 (day 1072). The top-performing models consistently identify a structural break near the second vaccine rollout.}
    \label{fig:model_ranking}
\end{figure}

\end{document}